\documentclass[reprint,superscriptaddress, aps]{revtex4-2}
\usepackage{amsmath,amssymb,amsthm,bm,amsfonts,mathrsfs,bbm}
\usepackage{dsfont}
\usepackage{graphicx}
\usepackage{dcolumn}
\usepackage{bm}
\usepackage{caption}
\usepackage{subcaption}

\newtheorem{definition}{Definition}
\newtheorem{corollary}{Corollary}
\newtheorem{example}{Example}
\usepackage{quantikz}
\usepackage{tikz}
\usetikzlibrary{decorations.pathreplacing}
\usepackage{ragged2e}
\usepackage{braket}
\newtheorem{theorem}{Theorem}
\newtheorem{lemma}{Lemma}

\usepackage{hyperref}
\newcommand{\lspan}{\texttt{span}}
\hypersetup{
    colorlinks=true,
    linkcolor=blue,
    citecolor=blue,
    urlcolor=blue
}
\usepackage{tikz-3dplot} 
\tdplotsetmaincoords{70}{120} 
\tdplotsetrotatedcoords{0}{0}{0} 
\usepackage{algorithm}
\usepackage[noend]{algpseudocode}
\algrenewcommand\algorithmicrequire{\textbf{Input:}}
\algrenewcommand\algorithmicensure{\textbf{Output:}}

\usepackage{booktabs}

\begin{document}

\preprint{APS/123-QED}

\title{Finding dense sub-lattices as low-energy states of a Hamiltonian}
\author{Júlia Barberà-Rodríguez}
\affiliation{ICFO-Institut de Ciències Fotòniques, The Barcelona Institute of Science and Technology,\\
Av. Carl Friedrich Gauss 3, 08860 Castelldefels (Barcelona), Spain}
\affiliation{SandboxAQ, Palo Alto, California, USA}
\author{Nicolas Gama}
\affiliation{SandboxAQ, Palo Alto, California, USA}
\author{Anand Kumar Narayanan}
\affiliation{SandboxAQ, Palo Alto, California, USA}
\author{David Joseph}
\affiliation{SandboxAQ, Palo Alto, California, USA}

\date{\today}

\begin{abstract}
Lattice-based cryptography has emerged as one of the most prominent candidates for post-quantum cryptography, projected to be secure against the imminent threat of large-scale fault-tolerant quantum computers. 
The Shortest Vector Problem (SVP) is to find the shortest non-zero vector in a given lattice. It is fundamental to lattice-based cryptography and believed to be hard even for quantum computers. 
We study a natural generalization of the SVP known as the $K$-Densest Sub-lattice Problem ($K$-DSP): to find the densest $K$-dimensional sub-lattice of a given lattice.    
We formulate $K$-DSP as finding the first excited state of a Z-basis Hamiltonian, making $K$-DSP amenable to investigation via an array of quantum algorithms, including Grover search, quantum Gibbs sampling, adiabatic, and variational quantum algorithms. The complexity of the algorithms depends on the basis through which the input lattice is presented. 
We present a classical polynomial-time algorithm that takes an arbitrary input basis and preprocesses it into inputs suited to quantum algorithms. 
With preprocessing, we prove that $O(KN^2)$ qubits suffice for solving $K$-DSP for $N$ dimensional input lattices. 
We empirically demonstrate the performance of a Quantum Approximate Optimization Algorithm $K$-DSP solver for low dimensions, highlighting the influence of a good preprocessed input basis.
We then discuss the hardness of $K$-DSP in relation to the SVP, to see if there is a reason to build post-quantum cryptography on $K$-DSP. 
We devise a quantum algorithm that solves $K$-DSP with run-time exponent $(5KN\log{N})/2$. Therefore, for fixed $K$, 
$K$-DSP is no more than polynomially harder than the SVP. 
The central insight we use is similar in spirit to those underlying the best-known classical $K$-DSP algorithm due to Dadush and Micciancio. Whether the exponential dependence on $K$ can be lowered remains an open question.  
\end{abstract}

\maketitle

\section{\label{sec:Intro}Introduction}
In 1994, Shor revolutionized quantum computing and cryptography by devising quantum algorithms that solve prime factorization and discrete logarithm in polynomial time~\cite{365700} using a fault-tolerant quantum computer. These problems are believed to be intractable on classical computers. Consequently, most public-key cryptographic protocols were built on such hardness assumptions~\cite{10.1145/359340.359342, Hankerson2004}. 

The emergence of fault-tolerant quantum computers that can run Shor's algorithm poses an imminent threat to public-key cryptography based on factoring or discrete logarithms. As a preventive measure, post-quantum cryptography (PQC) is being developed with urgency, spurred on by initiatives by the National Institute of Standards and Technology (NIST). Post-quantum cryptography is built on problems believed to be hard even on fault-tolerant quantum computers.
Within PQC, lattice-based cryptography (LBC)~\cite{Micciancio2009} is considered to be one of the most promising solutions to build quantum-safe primitives~\cite{alagic2020status}. This is also supported by complexity theoretic hardness reductions~\cite{10.1145/1568318.1568324}, as well as the failure of the best-known quantum algorithmic methods in breaking lattice-based problems. 

Public-key encryption schemes based on lattices were pioneered by Ajtai~\cite{10.1145/237814.237838}, and put on sound complexity theoretic footing by Regev~\cite{10.1145/1568318.1568324}. The Shortest Vector Problem (SVP) is to find the shortest non-zero vector of a given lattice. Regev based his schemes on the Learning With Errors (LWE) problem, the dual problem to the SVP.  The security of these systems is supported by a worst-case to average-case reduction from SVP to LWE-like problems. Thus, the hardness of SVP-like problems is the foundation of LBC. 
Since then, there has been a long line of LBC research leading to soon-to-be-standardized schemes such as Kyber, Dilithium, and Falcon~\cite{8406610, Ducas_Kiltz_Lepoint_Lyubashevsky_Schwabe_Seiler_Stehle_2018,fouque2018falcon}. 

The $K$-Densest Sub-lattice Problem ($K$-DSP) seeks the densest sub-lattice of a prescribed dimension $K$ within a given lattice. It is a generalization of the SVP since SVP is merely $1$-DSP. Being a generalization, $K$-DSP is at least as hard as the SVP but has received much less attention.  
While the primary focus of finding dense sub-lattices lies in cryptography~\cite{rankin_1955}, it also holds relevance in communication and information theory~\cite{4626061}, as well as in crystallography~\cite{gusev2023optimality}. 

The best-known classical algorithms for solving SVP start with classical reduction algorithms such as the LLL algorithm~\cite{lenstra1982factoring} and perform enumeration or sieving. Enumeration algorithms started with Pohst~\cite{10.1145/1089242.1089247} and Kannan ~\cite{10.1145/800061.808749},  and have evolved to be the best-known algorithms for SVP (see~\cite{eurocrypt-2010-23995} and the references therein).

Despite considerable interest, it is not clear if these sophisticated classical algorithms can be sped up meaningfully on a fault-tolerant quantum computer. However, the current interest is in designing quantum algorithms in the Noise Intermediate Scale Quantum (NISQ) era~\cite{Preskill_2018}, with limited qubits that are extremely sensitive to noise and very susceptible to decoherence effects. In this context, there is a necessity to develop hardware and quantum algorithms that are able to work properly even in the presence of noise~\cite{Bharti_2022}. 
Within this framework, Variational Quantum Algorithms (VQAs) emerge as one of the most encouraging approaches~\cite{Cerezo_2021}. These algorithms rely on running a parametrized quantum circuit on a quantum computer and perform classical optimization afterward to update the parameters of this circuit. 

The best-known classical algorithms for $K$-DSP are due to Dadush-Micciancio and for $K = \Theta(N)$ has a significantly worse complexity $K^{O({K N})}$ than SVP solvers~\cite{doi:10.1137/1.9781611973105.79}.
There have been previous works on mapping the SVP into one of finding low-energy states of a Hamiltonian.
Shorts vectors in the lattice map to eigenstates of the Hamiltonian with eigenvalues corresponding to the vector lengths. In~\cite{Joseph_2020}, the Bose-Hubbard Hamiltonian is used to solve the SVP via adiabatic quantum computation. Since their lowest energy state is the zero vector, the first excited state is the target of the algorithm. The results show that, outside the adiabatic regime and for up to 4-dimensional lattices, the solution for the SVP can be found with the same probability as the ground state or the second excited state. Two quantum algorithms are also proposed in~\cite{Joseph_2021}, where the authors consider the mapping of the SVP now into Ising coefficients which result in Hamiltonians that are more reusable in the context of gate model computers and certain annealer architectures.  
These algorithms are limited to low-dimensional lattices. A different approach with results for up to 28-dimensional lattices is presented in~\cite{Albrecht_2023}, employing the Ising spin Hamiltonian mapping followed by a Variational Quantum Eigensolver (VQE).  See~\cite{PhysRevA.106.022435, Joseph_2022, dableheath2023quantum, schirman2023finding} for additional quantum SVP-solving approaches.

\subsubsection{Contribution}
We construct a Hamiltonian from the input lattice whose eigenstates correspond to sub-lattices (of dimension at most $K$) with eigenvalues proportional to the covolumes of the sub-lattices. The ground energy eigenstates correspond to sub-lattices of dimension less than $K$. The solution to the $K$-DSP is the first excited states, corresponding to $K$-dimensional sub-lattices of the lowest covolume. Through this Hamiltonian correspondence, a variety of quantum algorithms ranging from adiabatic quantum computing to QAOA may be invoked to solve $K$-DSP. 
Such generality means that this Hamiltonian formulation can be useful in NISQ scenarios where VQA and QAOA-type approaches yield the best results, as well as in adiabatic systems, which may represent the most challenging engineering milestones in developing fault-tolerant quantum computation.    
For quantum algorithms that demand that the solutions to $K$-DSP be encoded as the lowest energy eigenstate, we present a way to penalize the ground states (fewer than $K$-dimensional sub-lattices) of the original Hamiltonian. To tune the penalization accurately, we resort to estimating the spectral gap of the Hamiltonian. 

The difficulty of lattice problems such as the SVP or $K$-DSP are not intrinsic to the input lattice but may depend on the basis through which the input lattice is presented. We present a classical polynomial-time algorithm that takes the input basis and preprocesses it to be suitable for being solved using quantum algorithms. In the simplest case, this preprocessing merely LLL-reduces the input basis and feeds it to the quantum algorithm. The LLL-reduced basis is a basis for the same input lattice but is better shaped to aid in the low-energy eigenstate search. In general, the preprocessing is far more intricate and may only call the quantum algorithm on a carefully chosen lower dimensional instance of $K$-DSP. With the preprocessing, we prove that $O(KN^2)$ qubits suffice for constructing the $K$-DSP Hamiltonian for $N$ dimensional input lattices with the assurance that the solution is contained in the span of the qubits (see Theorem~\ref{theorem:spatial}). A key technical insight that leads to the bound is the invariance of the densest sub-lattice under LLL-reduction. In particular, there is always a solution that is part of an LLL-reduced basis. With the number of qubits bounded, we may invoke quantum algorithms such as Grover search~\cite{grover1996fast}, phase estimation followed by amplitude amplification~\cite{kitaev1995quantum}, and quantum Gibbs sampling~\cite{chifang2023quantum} to find a low-energy state. These algorithms at the minimum come with a quadratic speedup over exhaustive search, but we may hope for more. The prospect of using quantum Gibbs sampling is particularly enticing since there is recent work proving that low-energy states of sparse (local) Hamiltonians drawn from random ensembles can be found efficiently~\cite{chifang2023sparse}. We do not expect cryptographic instances of lattice problems to fall within these ensembles, but further investigation is warranted before speculation. In particular, are there natural ensembles of lattice problems, perhaps arising in chemistry or digital communication for which quantum Gibbs sampling finds densest sub-lattices fast? The aforementioned quantum algorithms require fault-tolerant quantum computers. When constrained to NISQ devices, we can look to VQAs such as QAOA to find the low-energy states. We empirically investigate the performance of QAOA in small dimensions in Sec.~\ref{sec:exp}. 

Our phrasing of $K$-DSP as a problem of finding first excited states of a Hamiltonian is not complexity theoretically surprising, since $K$-DSP is in QMA. But when evaluating precise parameters to instantiate post-quantum cryptosystems, a more fine-grained complexity analysis and cryptanalysis tailored to the problem are necessary, beyond the general complexity class the problem lies in. The quadratic speedup achieved through our careful encoding and preprocessing can necessitate doubling the secret key, signature sizes, etc. of cryptosystems built on K-DSP.

We finally investigate the hardness of $K$-DSP as $K$ increases. Hardness cannot decrease with $K$, since $K_1$-DSP can be embedded into $K_2$-DSP for $K_2>K_1$ by appending short orthogonal basis vectors. But how much harder is $K$-DSP for large $K$, in particular, when $K$ approaches half the ambient dimension? If the best-known $K$-DSP algorithms take significantly longer than the SVP algorithms, then there is cause to base post-quantum cryptosystems on $K$-DSP. A natural way to address this question is to try to solve $K$-DSP given oracle access to an SVP solver. The best-known SVP solvers, classical or quantum, take exponential time in the lattice dimension. With an SVP solver, we observe that the input to quantum algorithms can be preprocessed to be an HKZ-basis, a stronger guarantee than an LLL-basis. Exploiting this stronger structure, we devise a quantum algorithm that solves $K$-DSP with run-time exponent $(5KN\log{N})/2$ (see Theorem~\ref{theorem:kdsp-groverization}). Therefore, for fixed $K$, $K$-DSP is no more than polynomially harder than the SVP. Our algorithm is close in spirit to the aforementioned best-known classical $K$-DSP algorithm by Dadush and Micciancio. In comparison, our run-time exponent is worse by a small constant. But our algorithm is simpler in the sense of not requiring recursive calls to the SVP oracle. Whether the exponential dependence on $K$ can be reduced remains an important open question.  

In Sec.~\ref{sec:preliminaries}, we introduce relevant theoretical concepts regarding lattices and quantum algorithms. In Sec.~\ref{sec:quantumalgo}, we describe the mapping of the Densest Sub-lattice Problem to one of finding low-energy states of a Z spins Hamiltonian. We also prove bounds on the quantum resources and then develop an approach to penalize trivial solutions by computing an upper bound on the spectral gap. In Sec.~\ref{sec:exp}, we simulate a $2$-DSP quantum solver and present the results obtained. Finally, Sec.~\ref{sec:conclusions} discusses the conclusions and potential directions for future research. 

\section{\label{sec:preliminaries}Preliminaries}
\subsection{Lattices}
A lattice is simply a pattern of repeating points in $N$-dimensional space, for example, a section of the integer lattice in 3 dimensions is visualized in Fig.~\ref{fig:DSP_representation}. More formally, we define a lattice as a discrete free $\mathbb{Z}$-sub-module of the real space $\mathbb{R}^N$ for some finite dimension $N$,  endowed with the Euclidean metric. The dimension of the lattice is its rank as a $\mathbb{Z}$-module, which is at most $N$.
A lattice of dimension $r\leq N$ can be described by an ordered basis $B=(\mathbf{b}_1,\mathbf{b}_2,\ldots,\mathbf{b}_r) \subset \mathbb{R}^N$ of $\mathbb{R}$-linearly independent vectors, presented in our algorithms as the matrix $\mathbf{B}$ consisting of basis vectors as the rows. Throughout, we represent vectors as row vectors. The lattice generated by $\mathbf{B}$ is the integer linear combinations $\mathcal{L}(\mathbf{B}):=\sum_{i=1}^N x_i \mathbf{b}_i = \{\mathbf{x} \mathbf{B}
\mid \mathbf{x}\in \mathbb{Z}^N\}$ of the basis vectors. 

The computational complexity of the problems we study is not intrinsic to the input lattice but may depend on the choice of basis generating the input lattice. One notion that helps solve these problems faster is that of a `good' basis.
Informally, a basis that contains only short and nearly orthogonal vectors is said to be a good basis. On the contrary, a bad basis consists of long vectors with large colinearities. A lattice has infinitely many bases, including a small number of good ones and a very large number of bad ones.

The LLL algorithm~\cite{lenstra1982factoring} is a powerful polynomial-time lattice reduction technique that transforms a given basis of a lattice into a ``better" basis for the same lattice, consisting of shorter basis vectors. Nevertheless, the basis returned is not good enough to solve the SVP. Such vectors that correspond to the output that LLL-reduction provides satisfy the Lovász condition. This involves a bound on the Gram-Schmidt basis lengths $(\mathbf{b}_i^*)$.

The fundamental parallelepiped of a basis $(\mathbf{B})$ is defined as $P(\mathbf{B}) := \{ \mathbf{x} \mathbf{B}| \mathbf{x} \in \left[0,1)\right.^N\}$. The covolume $\text{vol}\left(\mathcal{L}(\mathbf{B})\right):= \sqrt{\text{det}(\mathbf{B} \mathbf{B}^T)}$ of the fundamental parallelepiped only depends on the lattice $\mathcal{L}(\mathbf{B})$ and is invariant of the choice of the basis. It is known as the covolume of the lattice. The Gram matrix of a basis $\mathbf{B}$, is given by $\mathbf{G} = \mathbf{B} \mathbf{B}^T$ with $\mathbf{G}_{ij} = \mathbf{b}_i\mathbf{b}_j^T$ being the entries of the matrix.

One of the most widely studied computational problem on lattices is the Shortest Vector Problem or SVP.
\begin{definition}\textnormal{(The Shortest Vector Problem)}. 
Given a basis $\mathbf{B} = (\mathbf{b}_1,..., \mathbf{b}_N)$ of a lattice $\mathcal{L}(\mathbf{B})$, find the shortest non-zero vector 
\begin{equation}
\arg\min_z\{\lvert\lvert z \rvert\rvert_2 : z \in \mathcal{L}(\mathbf{B})\backslash \{0\}\}.
\end{equation}
with respect to the Euclidean norm. 
\end{definition}

Let $\lambda_1(\mathcal{L}) := \min_z\{\lvert\lvert z \rvert\rvert : z \in \mathcal{L}\backslash \{0\}\}$ denote the length of the shortest vector. Minkowski's theorem~\cite{minkowski1910geometrie} gives an upper bound on $\lambda_1(\mathcal{L})$. A Euclidean ball $S \subset \mathbb{R}^N$ of volume $\textnormal{vol}(S) > 2^N  \sqrt{\textnormal{vol}(\mathcal{L})}$ contains at least one lattice point in $\mathcal{L}$ that is not the origin. Then, $\lambda_1(\mathcal{L}) \leq \sqrt{N} \textnormal{vol}(\mathcal{L})^{2/N}$. Note that, the LLL-algorithm cannot find a vector within the Minkowski bound efficiently.

The SVP plays a significant role in PQC since it is known to be NP-hard under randomized reductions~\cite{10.1145/276698.276705}. Consequently, some primitives in LBC use short vectors. 
Since the SVP has been of interest to the cryptology community, many generalizations have arisen from it. 
One of them is the Densest Sub-lattice Problem, which seeks the densest $K$-dimensional sub-lattice of an arbitrary lattice. This problem was initially introduced in~\cite{10.1007/11818175_7}, and subsequently, Dadush and Miccancio developed the fastest classical algorithm for it~\cite{doi:10.1137/1.9781611973105.79}.
\begin{definition}\textnormal{(The Densest Sub-lattice Problem~\cite{10.1007/11818175_7})}. 
Given a basis $\mathbf{B} = ( \mathbf{b}_1,..., \mathbf{b}_N )$ that describes a lattice $\mathcal{L}$, and an integer $K$ such that $1\leq K\leq N$, find a $K$-dimensional sub-lattice $\mathcal{\hat{L}} \subseteq \mathcal{L}$ such that $\textnormal{vol}(\mathcal{\hat{L}})$ is minimal.
\end{definition}
Rankin's constants~\cite{rankin_1955} $\gamma_{N,K}(\mathcal{L})$ are a generalization of Minkowski's bound and are defined as 
\begin{equation}
    \gamma_{N,K}(\mathcal{L}) =\textnormal{sup}_\mathcal{L}\left ( \min \frac{\textnormal{det}((\mathbf{b}_1,...,\mathbf{b}_K))}{\textnormal{det}(\mathcal{L})^{K/N}} \right )^2.
    \label{eq:rankin}
\end{equation}
They can be interpreted as a bound on the densest sub-lattice covolume.

\begin{figure*}[t]
\includegraphics[width=1.5\columnwidth]{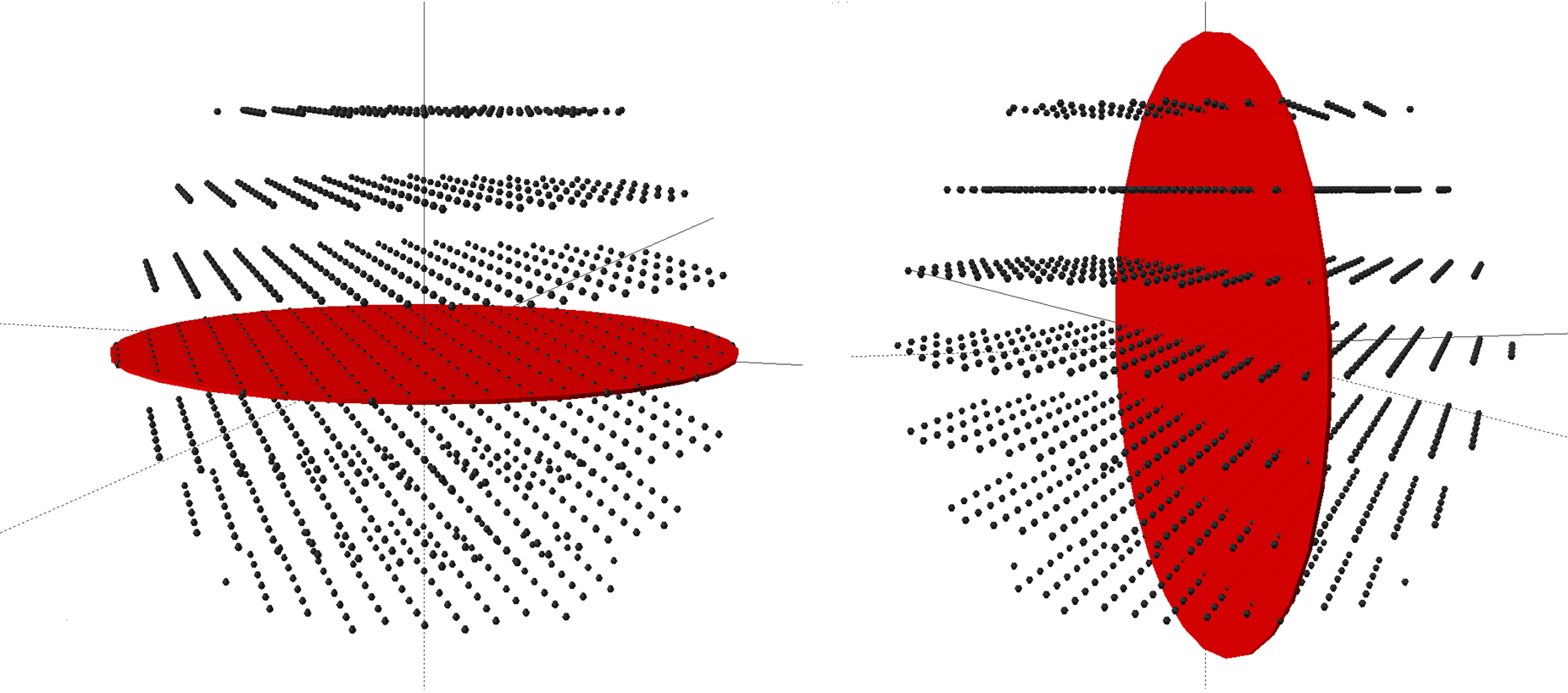}
\caption{The 3D lattice consists of sheets of stacked hexagonal honeycombs. The intersection of the red plane with the lattice defines a sub-lattice. The horizontal plane on the left carves a denser sub-lattice compared to the vertical plane on the right, evident from the larger number of intersection points in the red disk.}
\label{fig:kdsp_honeycomb}
\end{figure*}

In Fig.~\ref{fig:kdsp_honeycomb} we illustrate a case of the $K$-DSP for $K = 2$.  In a three-dimensional lattice, we examine two two-dimensional sub-lattices. Slicing along one axis and looking at the cross-section (red disk), reveals long thin rectangles of 1-meter width and 3-meter height. Slicing along another axis shows the honeycomb lattice with points arranged in equilateral triangles (1-meter sides). One can immediately see that the honeycomb lattice packs in many more points than the lattice of tall thin rectangles.

\begin{example}\textnormal{(DSP for $N = 3$ and $K = 2$). Assuming that we are given a $3$-dimensional input basis and that $K = 2$, the algorithm should return the two vectors that span the $2$-dimensional densest sub-lattice. This example is represented in Fig.~\ref{fig:DSP_representation}, where the red vectors represent the input and one solution to the problem could be given by the blue vectors. In this scenario, the input basis is $$ \mathbf{B}=\begin{bmatrix}
    1 & -1 & 0  \\
    0 & 1 & -1  \\
    0 & 0 & 1  
\end{bmatrix}$$ and a solution for the 2-DSP is given by $$\hat{\mathbf{B}}=\begin{bmatrix}
    1 & 0 & 0  \\
    0 & 1 & 0
\end{bmatrix}$$ embedded in the 2-dimensional subspace spanned by the $x,y$ axes. Here, the rows of the transformation matrix $\mathbf{X}$ are defined by $\mathbf{X} = \{(1,1,1),(0,1,1)\}$.}

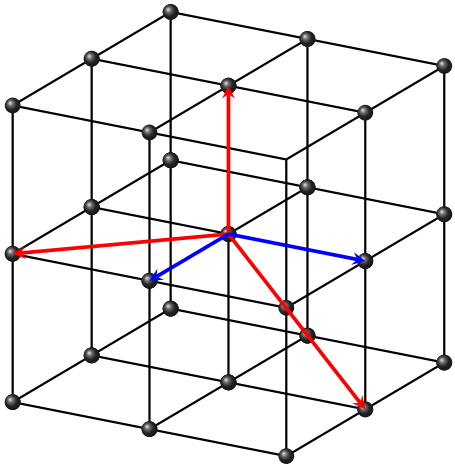
\begin{figure}[h]
\scalebox{0.65}{\begin{tikzpicture}[scale=3,tdplot_rotated_coords,
                    rotated axis/.style={->,purple,ultra thick},
                    blackBall/.style={ball color = black!80},
                    borderBall/.style={ball color = white,opacity=.25}, 
                    very thick]

\draw[->, thick] (1,1,1) -- (3,1,1) node[anchor=north east, scale = 2]{$x$};
\draw[->, thick] (1,1,1) -- (1,3,1) node[anchor=north west, scale = 2]{$y$};
\draw[->, thick] (1,1,1) -- (1,1,3) node[anchor=south, scale = 2]{$z$};

\foreach \x in {0,1,2}
   \foreach \y in {0,1,2}
      \foreach \z in {0,1,2}{

           \ifthenelse{  \lengthtest{\x pt < 2pt}  }{
             \draw[gray] (\x,\y,\z) -- (\x+1,\y,\z);
             \shade[rotated axis,blackBall] (\x,\y,\z) circle (0.05cm); 
           }{}

           \ifthenelse{  \lengthtest{\y pt < 2pt}  }{
               \draw[gray] (\x,\y,\z) -- (\x,\y+1,\z);
               \shade[rotated axis,blackBall] (\x,\y,\z) circle (0.05cm);
           }{}

           \ifthenelse{  \lengthtest{\z pt < 2pt}  }{
               \draw[gray] (\x,\y,\z) -- (\x,\y,\z+1);
               \shade[rotated axis,blackBall] (\x,\y,\z) circle (0.05cm);
           }{}

}

\shade[rotated axis,blackBall] (1,0,1) circle (0.05cm); 
\shade[rotated axis,blackBall] (1,2,1) circle (0.05cm); 
\shade[rotated axis,blackBall] (0,1,1) circle (0.05cm); 
\shade[rotated axis,blackBall] (2,1,1) circle (0.05cm);
\shade[rotated axis,blackBall] (1,1,2) circle (0.05cm);
\shade[rotated axis,blackBall] (1,1,0) circle (0.05cm);

\draw [-stealth, color = red, line width = 2pt](1,1,1) -- (2,0,1);
\draw [-stealth, color = red, line width = 2pt](1,1,1) -- (1,2,0);
\draw [-stealth, color = red, line width = 2pt](1,1,1) -- (1,1,2);

\draw [-stealth, color = blue, line width = 2pt](1,1,1) -- (2,1,1);
\draw [-stealth, color = blue, line width = 2pt](1,1,1) -- (1,2,1);
\end{tikzpicture}}
\caption{Representation of the 2-Densest Sub-lattice Problem for $N =3$. The red arrows represent a 3D input bad basis. One solution to the problem for $K = 2$ is given by the blue short vectors that span a 2D square lattice. The axis points in the positive direction.} 
\label{fig:DSP_representation}
\end{figure}

\textnormal{Notice that the 1-dimensional densest sub-lattice of any arbitrary lattice is by definition equivalent to solving the SVP. For instance, in Fig.~\ref{fig:DSP_representation}, the solution to the DSP for $K = 1$ would correspond to one of the blue vectors, which span a 1-dimensional sub-lattice. This is identical to solving the SVP within the 3D lattice. In this case, twice as many lattice points generated by (1,0,0) will fit in a given ball around the origin versus those generated by (2,0,0). Correspondingly, the $\sqrt{\text{det}(\mathbf{B} \mathbf{B}^T)}$ of the former sub-lattice is half that of the second. The other extreme, the $(N-1)$-DSP, is analogous to the SVP on a dual basis. In fact, for every $K$, the $K$-DSP restricted to $K$ is at least as hard as the SVP. Furthermore, we know that $K=N/2$ is the hardest instance since we can always add orthogonal vectors to the input basis and use a half-volume oracle to solve any $K$-DSP$_{K,N}$. The best-known algorithms for $K$-DSP with $K$ close to $N/2$ have $N^2$ in the runtime exponent, while SVP enumeration algorithms only have an $N$.}
\end{example}

 \begin{figure*}
\includegraphics[width=1.6\columnwidth]{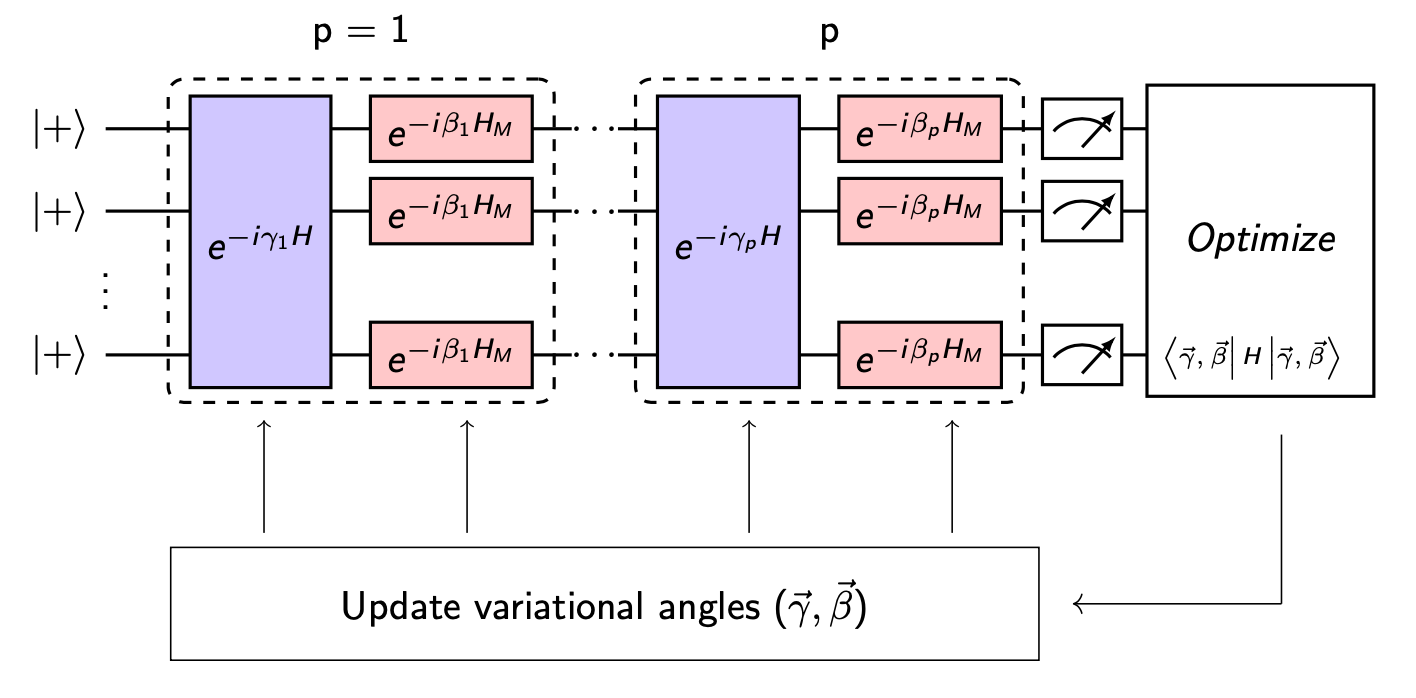}
\caption{Representation of the quantum circuit for the QAOA. The initial state is a superposition of all possible configurations. Then, $p$ layers are applied, each of them composed of the cost Hamiltonian $H$, which separates the states by their phase, and the rotation operator with the mixer Hamiltonian ($H_M$) in the exponent, transforming the phase into amplitude. In the end, the states of the qubits are measured, the output is post-processed, and the cost function is calculated using the initial parameters $(\boldsymbol{\gamma}, \boldsymbol{\beta})$. The angles are then updated, and the process repeats until a predefined stopping criterion, such as a target approximation ratio or stability in parameter updates, is satisfied.}
\label{fig:qaoa}
\end{figure*}
\subsection{Variational quantum algorithm}

Classical computers are unable to solve certain problems efficiently. For certain classes of problems, like period finding, we know large fault-tolerant quantum computers will offer speedups, and for other classes, we suspect they can offer speedups but have no guarantees. In the near term, however, these devices are not expected to be fault-tolerant. One approach to attempt quantum advantage in the NISQ era is Variational Quantum Algorithms (VQAs). The basic procedure for VQAs is the following:
\begin{enumerate}
    \item Prepare an initial state with $n$ available qubits $|0\rangle^n$.
    \item Apply a sequence of unitary gates 
    to construct the ansatz $|\psi(\boldsymbol{\theta})\rangle = U_p(\boldsymbol{\theta}_p),...,U_1(\boldsymbol{\theta}_1)|0\rangle^n$.
    \item Evaluate a cost function $C(\boldsymbol{\theta})$.
    \item Update parameters $\boldsymbol{\theta}$ via classical optimization to minimize the cost function $C(\boldsymbol{\theta})$. 
    \item Repeat the process many times until a stopping criterion is satisfied. The output of the algorithm is the parameterized quantum circuit defined by near-optimal parameters.
    \item Sample from the trained parameterized quantum circuit to obtain candidate solutions.
\end{enumerate} 
Picking an ansatz depends on a variety of factors, including resources available, the structure of the problem under consideration, and others. The cost function is designed so that it is minimized at the optimal solution, thus optimizing towards the minimum cost will result in a high probability of sampling good solutions.

However, it is important to remark on one of the main limitations of VQAs: Barren-Plateaus phenomenon~\cite{McClean_2018}. It consists of an effect produced when the gradient vanishes in the classical optimization process. As a consequence, the algorithm is not able to find the solution and the process can get trapped in local minima. 

One of the most widely used VQAs is the Quantum Approximate Optimization Algorithm (QAOA)~\cite{farhi2014quantum}. Given a cost Hamiltonian that encodes the solution of a combinatorial optimization problem in its ground state, the purpose of the algorithm is to return an approximate solution to the problem. It has been shown that for low-depth QAOA, speedups can be achieved in comparison to classical algorithms for certain instances~\cite{farhi2015quantum}. 
The QAOA Hamiltonian integrates the principles of the Transverse Ising model~\cite{Kadowaki_1998} 
and the $Z$-cost Hamiltonian. The Transverse Ising model represents a system of interacting spins in the presence of a transverse field induced by Pauli-$X$ operators. The spins are influenced by the Ising coupling term and the transverse field. Thus, the Hamiltonian can leverage the interaction between the spins and the transverse field to explore the solution space efficiently. In contrast, the $Z$-cost Hamiltonian drives the system towards optimal solutions by manipulating the quantum state based on the cost function. 

Concretely, in the case of QAOA for solving problems in the Transverse Ising model, one starts with an initial state $|\psi_0\rangle$, usually a superposition of all possible solutions as represented in Fig.~\ref{fig:qaoa}. Then, a sequence of two unitary operators $U(C, \gamma) = \exp{(-i\gamma H)}$ and  $U(B, \beta) = \exp{(-i\beta\sum_{i = 1}^n X_i)}$, generated by the cost Hamiltonian and the mixer Hamiltonian, are applied $p$ times. The parameterized quantum circuit simulation allows us to evolve the quantum system under the action of the two unitaries defined, $p$ times. Finally, the state of the qubits is measured, and the cost is updated. After each iteration, the angles $(\boldsymbol{\gamma}, \boldsymbol{\beta})$ are adjusted, and all the steps are repeated until convergence of the cost function.
When increasing $p$, the quality of the approximation of the final state to the ground state of the cost Hamiltonian is enhanced. In the limit $p \rightarrow \infty$, the algorithm can emulate the adiabatic quantum algorithm, achieving a high overlap with the ground state, provided that the classical optimizer finds near-optimal parameters. 

\section{\label{sec:quantumalgo}A quantum algorithm for the Densest Sub-lattice Problem}

In~\cite{doi:10.1137/1.9781611973105.79} it is shown that if $\mathcal{L}$ is an $N$-dimensional lattice, $\mathcal{\hat{L}} \subseteq \mathcal{L}$ a $K$-dimensional sub-lattice of minimum determinant and $\mathbf{v} \in \mathcal{L}$ be any lattice vector, then, either $\mathcal{\hat{L}}$ contains all shortest lattice vectors or the length of all vectors that span the densest sub-lattice is bounded by $\lambda_K(\mathcal{\hat{L}}) \leq K\lambda_1(\mathcal{L})$. This lemma enables to develop a classical enumerative algorithm for the $K$-DSP with exponential running time $K^{O(K N)}$. 

In this section, we derive the mapping of the Densest Sub-lattice Problem to a $Z$ spins Hamiltonian such that its first excited state $E_1$ corresponds to the densest sub-lattice of a given ambient lattice. We also provide the space requirements for the quantum algorithm, as well as an approach to penalize trivial solutions of the $K$-DSP by upper-bounding the spectral gap. 

\subsection{Hamiltonian formulation for the $K$-DSP}
Given a basis $\mathbf{B}$ that describes a $N$-dimensional lattice $\mathcal{L}$, we seek $K$-linearly independent vectors $\mathbf{v}_1, ... , \mathbf{v}_K$ that span a sub-lattice $\mathcal{\hat{L}} \subseteq \mathcal{L}$ with the smallest determinant. 

Let us consider that $K = 2$, for simplicity, and that $N$ can take any value. Nevertheless, one can see App.~\ref{app:appendixA} for the formulation of the $K$-DSP for $ K = 3$. Then, the goal is to find two linearly independent vectors that span the densest sub-lattice, or equivalently, find the $2$-dimensional sub-lattice with the smallest covolume or determinant. Thus, we can express these two lattice points as a linear combination of input basis vectors such that
\begin{eqnarray}\label{equation_xy}
     \mathbf{v}_1 &= \mathbf{x} \mathbf{B} = x_1  \mathbf{b}_1 + ... + x_N  \mathbf{b}_N, \nonumber\\ 
     \vspace{5mm}
    \mathbf{v}_2 &= \mathbf{y}\mathbf{B} = y_1  \mathbf{b}_1 + ... + y_N  \mathbf{b}_N.  
    \label{eq:vectors}
\end{eqnarray}
Recall that the Gram matrix is a square matrix for any $K$ and $N$. The Gramian, which is the determinant of the Gram matrix, is equal to the covolume squared of a lattice $\mathcal{L}$,
\begin{equation}
   \text{vol}(\mathcal{L}) = \sqrt{\text{det}(\mathbf{G(\mathbf{b}_1,..., \mathbf{b}_N)})}.
\end{equation}
Therefore, the covolume of a sub-lattice $\mathcal{\hat{L}}$ of dimension $K = 2$ of an $N$-dimensional ambient lattice $\mathcal{L}$ can be expressed as 
\begin{equation}
   \text{vol}(\mathcal{\hat{L}})^2 = 
   \begin{vmatrix}
       \langle \mathbf{v}_1 | \mathbf{v}_1 \rangle &  \langle \mathbf{v}_1 | \mathbf{v}_2 \rangle \\
       \langle \mathbf{v}_2 | \mathbf{v}_1 \rangle &  \langle \mathbf{v}_2 | \mathbf{v}_2 \rangle
   \end{vmatrix}.
\end{equation}

Using Eq.~(\ref{eq:vectors}) in the former determinant, the following relation is obtained between the covolume of the sub-lattice and the coefficients $\mathbf{x}$ and $\mathbf{y}$,
\begin{equation}
   \text{vol}(\mathcal{\hat{L}})^2 = \sum_{i,j,k,l}^N x_i x_j y_k y_l (\mathbf{G}_{ij} \mathbf{G}_{kl} - \mathbf{G}_{ik} \mathbf{G}_{jl}).
    \label{eq: cost}
\end{equation}
\noindent
The different inner products $\mathbf{b}_i\mathbf{b}_j$ have been written as $\mathbf{G}_{ij}$ and therefore, are constants given as inputs of the algorithm. 

The main goal of the problem is to find the integer variables $x_i$ and  $y_j$ such that they minimize Eq.~(\ref{eq: cost}). This equation simulates the cost function and returns the different covolumes that can be found within the system. To transform the equation into a Hamiltonian we have to consider that the eigenvalues must correspond to the energies of the system, which in this context are the squared covolumes. Then, the eigenvectors of this Hamiltonian are given by the different sub-lattices that can be constructed. In this way, when applying the problem Hamiltonian over an eigenstate (i.e. a sub-lattice) the energy returned will be the covolume squared of the configuration considered. 
Note that, to be able to obtain the integer values that define the coefficients in Eq.~(\ref{eq: cost}), the integer coefficients need to be modified to binary variables. 

In~\cite{Joseph_2021} the authors propose a qudit mapping that will be useful for the simulation of our algorithm. The Binary-encoded qudits mapping allows to interpret a binary string of qubits as integers. Let us assume that we have four qubits per qudit available as in Fig.~\ref{fig:grid}. Then, the Hilbert space of the qudit operator will be spanned by $2^4$ states.

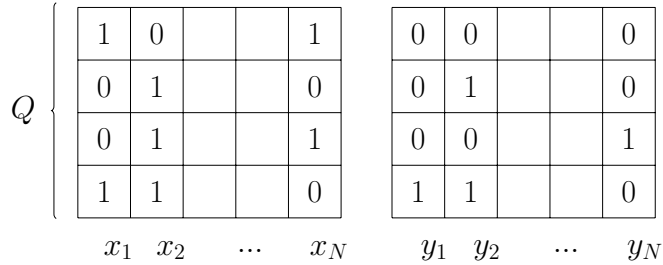
\begin{figure}[ht]
\LARGE{
\begin{flushright}
\scalebox{0.7}{\begin{tikzpicture}[remember picture]
\draw[step=1cm,color=black] (-1,-1) grid (4,3);
\node at (-0.5,+2.5) {1};
\node at (-0.5, 1.5) {0};
\node at (-0.5,0.5) {0};
\node at (-0.5,-0.5) {1};

\node at (0.5,+2.5) {0};
\node at (0.5, 1.5) {1};
\node at (0.5, 0.5) {1};
\node at (0.5,-0.5) {1};

\node at (+3.5,+2.5) {1};
\node at (+3.5,+1.5) {0};
\node at (+3.5, 0.5) {1};
\node at (+3.5, -0.5) {0};
\end{tikzpicture}}
\hspace{3mm}
\scalebox{0.7}{\begin{tikzpicture}
\draw[step=1cm,color=black] (-1,-1) grid (4,3);
\node at (-0.5,+2.5) {0};
\node at (-0.5, 1.5) {0};
\node at (-0.5,0.5) {0};
\node at (-0.5,-0.5) {1};

\node at (0.5,+2.5) {0};
\node at (0.5, 1.5) {1};
\node at (0.5, 0.5) {0};
\node at (0.5,-0.5) {1};

\node at (+3.5,+2.5) {0};
\node at (+3.5,+1.5) {0};
\node at (+3.5, 0.5) {1};
\node at (+3.5, -0.5) {0};
\end{tikzpicture}}

\scalebox{0.7}{\begin{tikzpicture}[remember picture, overlay]

\draw [decorate,
    decoration = {brace}] (-11.2,1.1) --  (-11.2,5.2) node[midway, left]{$Q\hspace{2mm}$}; 
\end{tikzpicture}
\begin{tikzpicture}[remember picture, overlay]
\node at (-10.2,0.5) {$x_1$};
\node at (-9.2,0.5) {$x_2$};
\node at (-7.7,0.5) {$...$};
\node at (-6.2,0.5) {$x_N$};
\node at (-4.2,0.5) {$y_1$};
\node at (-3.2,0.5) {$y_2$};
\node at (-1.7,0.5) {$...$};
\node at (-0.2,0.5) {$y_N$};
\end{tikzpicture}}
\end{flushright}}
\caption{Diagram representing two $4\times N$ dimensional grids of qubits. Each of them is associated with one of the vectors that span the sub-lattice of dimension $K$. The columns of the 2D array represent the integer values that can take each of the coefficients $x_i$ and $y_j$. The columns, each of which represents a qudit, are composed of four qubits.}
\label{fig:grid}
\end{figure}
The measurements are performed in the computational basis $Z$, so return the eigenvalue $+1(-1)$. Since we would like to work with binary strings, we can transform the $Z$-basis to the $(0,1)$ basis, via the operator
\begin{equation}
    \hat{O} = \frac{\mathds{1}- Z}{2}.
    \label{eq:operator}
\end{equation}
 
From $\hat{O}$ we can now construct the qudit operator
\begin{equation}
    \hat{Q}_{bin}^{(i)} = \sum_{w = 0}^{m} 2^{w}\hat{O} - 2^{m} \mathds{1},
    \label{eq:mapping}
\end{equation}
where the first term contributes by returning the integer associated to the $i^{th}$ qudit state when the operator is applied, and the second term is used to shift the range down to be symmetric about zero and thus allow for negative coefficients. Rearranging Eq.~(\ref{eq:mapping}) using Eq.~(\ref{eq:operator}) we obtain
\begin{equation}
\hat{Q}_{bin}^{(i)} = -\frac{\mathds{1}}{2}-\sum_{w = 0}^{m} 2^{w-1} Z_{wi},
\label{eq:mapping2}
\end{equation}
which translates the measured qubits of the grid's columns to integers in the range $[-2^{m}, 2^{m}-1]$. 

Thus, the problem Hamiltonian for the $2$-DSP results in
\begin{equation}
   H_{2DSP} = \sum_{i,j,k,l = 1}^N \hat{Q}^{(i)}\hat{Q}^{(j)}\hat{Q}^{(k)}\hat{Q}^{(l)} (\mathbf{G}_{ij}\mathbf{G}_{kl}-\mathbf{G}_{ik}\mathbf{G}_{jl})
   \label{eq:ham_pen}
\end{equation}
which comes from Eq.~(\ref{eq: cost}) when the coefficients $\boldsymbol{x}, \boldsymbol{y}$ are replaced by the qudits operators.
Expanding the Hamiltonian for $K = 2$ we obtain 
\begin{align}
&H_{2DSP} = \sum_{i,j,k,l = 1}^N 2^{-4}\left[\sum_{a,b,c,d=0}^{m} 2^{a+b+c+d} Z_{ai}Z_{bj}Z_{ck}Z_{dl}  \right. \nonumber \\
&\left.+\sum_{a,b,c =0}^{m} 2^{a+b+c} (((Z_{ai}+Z_{aj})Z_{bk} + Z_{ai}Z_{bj})Z_{cl}  \right.\nonumber \\
& \left. + Z_{ai}Z_{bj}Z_{ck}) + \sum_{a,b =0}^{m} 2^{a+b}((Z_{ai} + Z_{aj} + Z_{ak})Z_{bl}  \right.\nonumber \\
& \left. + (Z_{ai}+ Z_{aj})Z_{bk} + Z_{ai}Z_{bj} + \sum_{a = 0}^{m} 2^{a}(Z_{ai} + Z_{aj} \right.\nonumber \\
& \left.+ Z_{ak} + Z_{al}) + 1
\right] (\mathbf{G}_{ij}\mathbf{G}_{kl}-\mathbf{G}_{ik}\mathbf{G}_{jl}),
\label{eq:z-spins-ham}
\end{align}

\noindent
which corresponds to a fully connected Hamiltonian, and has up to 4-body interactions. The Eq.~(\ref{eq:z-spins-ham}) is the one that will be used for the simulation of the quantum $2$-DSP solver in Sec.~\ref{sec:exp}. 

To generalize the problem Hamiltonian to any $K$, consider the $K$-DSP for a $N$-dimensional lattice spanned by $\mathbf{B} = (\mathbf{b}_1,...,\mathbf{b}_N)$. The goal is to find $K$-linearly independent vectors $\hat{\mathbf{{B}}} = (\mathbf{v}_1,...,\mathbf{v}_K)$, which generate $\mathcal{\hat{L}}$ with the smallest covolume. 
Here, the Gramian of $\hat{\mathbf{{B}}}$ is given by 
\begin{equation}
\textnormal{det}(\mathbf{G}(\mathbf{v}_1,...,\mathbf{v}_K)) = 
   \begin{vmatrix}
\langle\mathbf{v}_{1}, \mathbf{v}_{1} \rangle&\hdots & \langle\mathbf{v}_{1}, \mathbf{v}_{K} \rangle \\
\vdots & \ddots & \vdots \\
\langle\mathbf{v}_{K}, \mathbf{v}_{1} \rangle &\hdots & \langle\mathbf{v}_{K}, \mathbf{v}_{K} \rangle
\end{vmatrix}.
\end{equation}
The determinant of a $K\times K$ square matrix $\mathbf{A}$ with entries $a_{ij}$, can be computed using Leibniz formula
\begin{equation}
    \textnormal{det}(\mathbf{A}) = \sum_{\tau \in S_K} sgn(\tau) \prod^K_{i=1} a_{i,\tau(i)},
\end{equation}
where the sum runs over all the permutations $\tau$ of the symmetric group $S_N$, and $sgn \in \{ \pm 1\}$ is the sign of $\tau$. One can write the generalized covolume of a $K$-dimensional sub-lattice in terms of the determinant of the rank $K$ Gram matrix, formed by the $K$ lattice vectors returned

\begin{equation}
    \textnormal{vol}(\hat{\mathcal{L}})^2 = \sum_{\tau \in S_K} sgn(\tau) \prod^K_{i=1} \langle\mathbf{v}_{i}, \mathbf{v}_{\tau(i)}\rangle. 
\end{equation}

To encode the Hamiltonian of the general $K$-DSP, one can write the vectors $\mathbf{v}_i$ in terms of the qudit operators
\begin{equation}
    H_{DSP} =\sum_{\tau \in S_K} sgn(\tau) \prod^K_{i=1} \left (\sum^N_{\alpha, \beta = 1}\hat{Q}^{(i)}_{\alpha}\hat{Q}^{(\tau(i))}_{\beta} \mathbf{G}_{\alpha, \beta}\right)
\end{equation}
where the operators $\hat{Q}^{(i)}_{\alpha}$ and $\hat{Q}^{(\tau(i))}_{\beta}$ act on the qudits $\alpha$ and $\beta$ within a register, and $i, \tau(i)$ index the qubit grids. Each qubit grid describes a single vector, two of which are shown in Fig.~\ref{fig:grid}. Here, $\mathbf{G}$ is the Gram matrix of the full $N$-dimensional input basis $\mathbf{B}$. The output of this operation consists of two integer values associated with the classical coefficients $x_{i,\alpha}$ and $x_{\tau(i),\beta}$ that generate the solution vectors $\mathbf{v}_i$ and $\mathbf{v}_{\tau(i)}$. The eigenenergies of $H_{DSP}$ are then the $\textnormal{vol}(\hat{\mathcal{L}})^2$ where the $\hat{\mathcal{L}}$ is the sub-lattice corresponding to the relevant eigenvector.

One appeal of using the Leibnitz formula for the determinant is that it gives Hamiltonians that are 2$K$-local. Since even 3-local Hamiltonians are QMA-hard~\cite{qma_hard}, in the worst case, locality does not help in the computation of low-energy states. However, for local Hamiltonians drawn from certain natural random ensembles, locality does seem to help~\cite{chifang2023sparse}. To efficiently implement the Hamiltonian, the summation over the symmetric group is done in superposition. The operators inside the summation indexed by permutations are implemented using conditional gates. An alternative to the Leibnitz formula is to use an efficient arithmetic circuit to compute the determinant, such as the one derived through Gaussian elimination. 

\subsection{Classical preprocessing and spatial bound}
The coordinates of the unknown vectors (such as $\mathbf{x}$ and $\mathbf{y}$ in Eq.~(\ref{equation_xy})) in the Hamiltonian formulation are apriori unbounded integers. Therefore, the naive search space is countably infinite. To translate the Hamiltonian formulation into a quantum algorithm, it is necessary to confine the solution space to be finite, preferably small. In particular, this bounds the number of qubits required by the algorithm. For the SVP, the Minkowski bound readily confines the search space. For $K$-DSP with $K>1$, bounding the search space is far more intricate. We devise a polynomial-time classical preprocessing algorithm to resolve this problem.

The classical preprocessing algorithm takes a lattice basis $\mathbf{B}_{in}$ as the input and outputs a preprocessed gap-free basis $\mathbf{B}_{P}$ in polynomial time. First, it runs the LLL algorithm on $\mathbf{B}_{in}$ to get another basis $\mathbf{B}_{L}$ generating the same lattice. This $\mathbf{B}_{L}$ either has a gap or is gap-free. This notion of a Gap is one that we define and can be tested in polynomial time \footnote{An alternative characterization of gap that can be found in the literature is between the successive minima of the lattice when there exists an index $p$ s.t $\lambda_{p+1}(L)\gg\lambda_{p}(L)$. A lattice that has a gap larger than $O(N^{1/2})$ according to one definition has a gap also with the other definition and vice versa.}
. It is not intrinsic to a lattice and may depend on the basis. 
\begin{definition}\label{definition_gap}\textnormal{(Gap).}
A basis $\mathbf{B}=(\mathbf{b}_1,\ldots,\mathbf{b}_N)$ is defined to have a Gap if there exists an index $r$ such that, $\max(||\mathbf{b}_1^*||,...,||\mathbf{b}_r^*||) < \min(||\mathbf{b}_{r+1}^*||,...,||\mathbf{b}_N^*||)$,
where $(\mathbf{b}_1^*, \mathbf{b}_2^*,\ldots,\mathbf{b}_N^*)$ is the result of the Gram-Schmidt orthogonalization obtained from $\mathbf{B}$.
\end{definition}

If $\mathbf{B}_{L}$ is a gap-free basis, the preprocessing algorithm outputs $\mathbf{B}_{L}$ as $\mathbf{B}_{P}$. This is the typical case since generic lattices are gap-free as a consequence of the Gaussian heuristic: the presence of a gap would mean that at least one of its projected lattices has a shortest vector much below the Gaussian heuristic, which is highly improbable. The quality of the preprocessed basis $\mathbf{B}_{P}=\mathbf{B}_{L}$ is quantified in Lemma~\ref{lem:gap-free-lll}.  
\begin{lemma}[Gap-free-LLL]\label{lem:gap-free-lll}
If a LLL-reduced basis $\mathbf{B}$ of $\mathcal{L}$ is gap-free, then $||\mathbf{B}^*||\leq (4/3 +\varepsilon)^{(N-1)/4}\textnormal{vol}(\mathcal{L})$.
\end{lemma}
\begin{proof}
Let $x_i$ denote the values of the $N$ norms $\mathbf{b}_j^*$ sorted by increasing order. Then if there exists an index $p$ s.t. $x_{p+1}>\sqrt{4/3} x_{p}$, let us paint in blue the indexes s.t. $\|\mathbf{b}_j^*\|\leq x_p$ and in red those where $\|\mathbf{b}_j\| \geq x_{p+1}$. Because of Lovász conditions~\cite{lenstra1982factoring}, a red index cannot be followed by a blue index in the basis, so the only viable coloring is to have all the blue indexes followed by all the red indexes. Then necessarily, $(\|\mathbf{b}_1^*\|, \dots, \|\mathbf{b}_p^*\|)$ are all $\leq x_p$ and $(\|\mathbf{b}_{p+1}^*\|, \dots, \|\mathbf{b}_N^*\|)$ are all $\geq x_{p+1}$, and the basis has a gap.
Reciprocally, a gap-free LLL-reduced basis satisfies $x_i \leq x_{i+1} \leq \sqrt{4/3} x_i$ for all $i \in [1, N-1]$, and therefore, $x_N/(\prod x_j)^{1/N} = \max{\mathbf{b}_j^*}/\textnormal{vol}(\mathcal{L})^{1/N} \leq \left(\frac{4}{3}\right)^{(N-1)/4}$.
\end{proof}

If $\mathbf{B}_{L}$ is not a gap-free basis, we invoke Lemma~\ref{lem:gap-reduction} (which builds on Lemma~\ref{lem:sub-lattice-hnf}) to reduce the dimension of the problem in polynomial time. In particular, Lemma~\ref{lem:gap-reduction} identifies a gap-free LLL-reduced basis as the preprocessed basis $\mathbf{B}_{P}$ to be fed as the input to a $K$-DSP problem in fewer dimension. Lemma~\ref{lem:relative-basis-size} quantifies the quality of $\mathbf{B}_{P}$.

\begin{lemma}[Dual of sub-lattice]\label{lem:sub-lattice-hnf}Let $\mathcal{L}$ be a lattice and $\mathbf{B}$ its basis, and $\hat{\mathcal{L}}\subseteq \mathcal{L}$ a sub-lattice of dimension $D$. Let $K$ be the smallest index s.t. $\hat{\mathcal{L}}\subseteq L(\mathbf{b}_1,\dots,\mathbf{b}_K)$, then there exists a basis $\mathbf{C}=(\mathbf{c}_1,\dots,\mathbf{c}_D)$ of $\hat{\mathcal{L}}$ such that $\|\mathbf{c}_D^*\|\geq\|\mathbf{b}_K^*\|$.
\end{lemma}
\begin{proof}
There exists a $D\times K$ integer matrix $\mathbf{V}$ s.t. $\mathbf{V}(\mathbf{b}_1,\dots,\mathbf{b}_K)$ is a basis of $\hat{\mathcal{L}}$, furthermore, the last column of $\mathbf{V}$ is non-zero. Let $\mathbf{V}'$ be the Hermite normal form of $\mathbf{V}$, then $\mathbf{C}=\mathbf{V}'(\mathbf{b}_1,\dots,\mathbf{b}_K)$ is still a basis of $\hat{\mathcal{L}}$, and $\mathbf{V}'$ has a single non-zero integer coefficient in its last column, in position $(D,K)$. It follows that $\lspan(\mathbf{c}_1,\dots,\mathbf{c}_{D-1})\subseteq \lspan(\mathbf{b}_1,\dots,\mathbf{b}_{K-1})$ and $\mathbf{c}_K = \mathbf{v}_{D,K}\mathbf{b}_K + \mathbf{t}$ where $\mathbf{t}\in\lspan(\mathbf{b}_1,\dots,\mathbf{b}_{K-1})$. Therefore, $\|\mathbf{c}_K^*\| \geq \|\mathbf{v}_{D,K}\mathbf{b}_K^*\| \geq \|\mathbf{b}_K^*\|$. 
\end{proof}

\begin{lemma}[Gap-dimension-reduction]\label{lem:gap-reduction}
If $\mathbf{B}$ is a basis of $\mathcal{L}$ and has a gap, then the $K$-DSP on $\mathbf{B}$ reduces to a lower dimension for all $K$
\end{lemma}
\begin{proof}
Assume that the input basis has a gap at index $p$ and call $T=\max_{i\in[1,p]}{||\mathbf{b}_i^*||}$ and $\mathbf{E}=\lspan(\mathbf{b}_1,\dots,\mathbf{b}_p)$. Let $(\mathbf{c}_1,\dots,\mathbf{c}_K)$ be a dual-HKZ-reduced basis~\cite{korkine1877formes} of a sub-lattice with minimal covolume, $\mathbf{F}=\lspan(\mathbf{c}_1,\dots,\mathbf{c}_K)$, and let $\pi_{K-1}$ be the projection on the orthogonal of $\lspan(\mathbf{c}_1,\dots,\mathbf{c}_{K-1})$.

If $K\leq p$, we will prove that $\mathbf{F}\subseteq \mathbf{E}$. Suppose by contradiction that this inclusion does not hold.
In this case, $\|\mathbf{c}_K^{*}\|$ is maximal across all bases, so by Lemma~\ref{lem:sub-lattice-hnf}, it is $>T$. Also, at least one of the vectors $\mathbf{b}_1,\dots,\mathbf{b}_p$ is not in $\mathbf{F}$: it can either be a dimension argument if $p>K$, or due to the non-inclusion assumption when $p=K$. So at least one of the $\pi(\mathbf{b}_i)$ is non-zero. Let $j$ be the smallest of such index, we therefore have $\|\pi(\mathbf{b}_j)\|\leq \|\mathbf{b}_j^*\| \leq T$. If we replace $\mathbf{c}_K$ with $\mathbf{b}_j$ in the basis $\mathbf{C}$, we would obtain a sub-lattice of shorter covolume, which contradicts that $\mathbf{C}$ is a $K$-DSP solution. Therefore, $\mathbf{F}\subseteq \mathbf{E}$.
If $K=p$, because $\mathbf{E}$ and $\mathbf{F}$ have also the same dimension, they are equal, and $\mathbf{b}_1,\dots,\mathbf{b}_p$ is a solution of the $K$-DSP.
If $K<p$, we just proved that the $K$-DSP solution is a sub-lattice of $(\mathbf{b}_1,\dots,\mathbf{b}_p)$, so we reduce the problem to a smaller dimension.
If $K>p$, we prove by duality that $\mathbf{E}\subset \mathbf{F}$, so by projecting the input basis on the orthogonal of $\mathbf{E}$, we are reduced to solve the smaller $(K-p)$-DSP on the $(N-p)$-dimensional projected basis.
\end{proof}

\begin{lemma}[Relative-basis-size]\label{lem:relative-basis-size}
If $\mathbf{B}$ and $\mathbf{C}$ are two bases of $\mathcal{L}$, and $\mathbf{C}$ has a gap at index $K$, then $\textnormal{vol}(\pi_K(\mathbf{C}))^{1/(N-K)} \leq ||\mathbf{B}^*||$ where $\pi_K$ is the orthogonal projection over $(\mathbf{c}_1,\dots,\mathbf{c}_{K-1})^\perp$.
\end{lemma}
\begin{proof}
By definition, the family $\mathbf{F}=\pi_K(\mathbf{B})$ generates the same lattice as $\pi_K(\mathbf{C})$, and $\|\mathbf{F}^*\|\leq \|\mathbf{B}^*\|$. If we use the LLL algorithm on $\mathbf{F}$, we obtain a basis $\mathbf{F'}$ of $\pi_K(\mathbf{C})$ with $\|\mathbf{F'}^*\|\leq \|\mathbf{B}^*\|$. In particular, $\textnormal{vol}(\pi_K(\mathbf{C}))^{1/N-K}$ is the geometric mean of $\|\mathbf{F'}_i^*\|$, so $\textnormal{vol}(\pi_K(C))^{1/N-K}\leq \|\mathbf{F'}^*\|\leq \|\mathbf{B}^*\|$ 
\end{proof}

\begin{corollary}\label{cor:relative-lll-basis-size}
If $\mathbf{B}$ and $\mathbf{C}$ are two bases of $\mathcal{L}$ and $\mathbf{C}$ is LLL-reduced, then $||\mathbf{C}^*|| \leq (4/3 +\varepsilon)^{(N-1)/4}||\mathbf{B}^*||$.
\end{corollary}
\begin{proof}
This is a consequence of the previous Lemma~\ref{lem:relative-basis-size} when applied to the highest gap of $\mathbf{C}$ if $\mathbf{C}$ has a gap, or the whole basis if $\mathbf{C}$ is gap-free. 
\end{proof}

In the proof of Theorem~\ref{theorem:spatial}, a key insight is that there is an output basis that is a solution to $K$-DSP that is LLL-reduced, due to the invariance of $K$-DSP solutions under LLL-iterations. Furthermore, irrespective of whether the original input basis $\mathbf{B}_{in}$ is gap-free, the output $\mathbf{B}_{P}$ of the preprocessing is always LLL-reduced and gap-free. Therefore, both the input basis $\mathbf{B}_{P}$ and the output basis $\mathbf{B}_{out}$ of the quantum algorithm can be constrained to be LLL-reduced without loss of generality. This structure, in concert with the quality assurance of the input basis $\mathbf{B}_{P}$ from Lemmas~\ref{lem:gap-free-lll} and~\ref{lem:relative-basis-size}, allows us to bound the qubits.

\begin{lemma}{\textnormal{(Unimodular-transformation-bound)}}\label{lemma:unimodular_bound}
If $\mathbf{B}_P$ is a gap-free LLL basis of $\mathcal{L}$ and $\mathbf{C}$ is an LLL-reduced basis of the same lattice, then the unimodular transformation $\mathbf{U}$ s.t. $\mathbf{U}\mathbf{B}_P=\mathbf{C}$ satisfies $\|\mathbf{U}\|_\infty \leq N(4/3 +\varepsilon)^{3(N-1)/4}.$
\end{lemma}

\begin{proof}
Since $\mathbf{B}_P$ is an LLL-reduced gap-free basis, so is its (reversed) dual $\mathbf{B}_P^{-t}$. Therefore, by Lemma~\ref{lem:gap-free-lll}, we have $\|\mathbf{B}_P^*\| \leq (4/3 +\varepsilon)^{(N-1)/4}\textnormal{vol}(\mathcal{L})$ and $\|\mathbf{B}_P^{*-t}\| \leq (4/3 +\varepsilon)^{(N-1)/4}\textnormal{vol}(\mathcal{L})^{-1}$. By Corollary~\ref{cor:relative-lll-basis-size}, $\|\mathbf{C}^*\|\leq (4/3 +\varepsilon)^{(N-1)/4}\|\mathbf{B}_P^*\|$. We prove the Lemma by converting the Gram-Schmidt norm to the spectral norm and multiplying them together.
\end{proof}
The previous Lemma~\ref{lemma:unimodular_bound} allows us to bound the size of the unimodular transformation $\mathbf{U}$ from the input basis to the output basis. The entries of the matrix $\mathbf{U}$ consist of qubits, which are the target coefficients we are seeking. Therefore, determining an upper limit on these entries will enable us to bound the total number of qubits required for the quantum solver.

\begin{theorem}\textnormal{(Bound on the number of qubits).}\label{theorem:spatial}
Let $N$-dimensional lattice $\mathcal{L}$ be the input to the quantum algorithm for the $K$-DSP as the span of a gap-free LLL-reduced basis $\mathbf{B}_P= (\mathbf{b}_1,..., \mathbf{b}_N)$. Then, a $KN^2$ qubit Hilbert search space is sufficient to ensure that at least one exact solution of the $K$-DSP is contained. 
\end{theorem}

\begin{proof} The worst case in terms of the number of qubits will occur when the input basis has no gaps since the dimension of the problem cannot be reduced. Let the input of the quantum solver be a gap-free LLL-reduced basis of dimension $N$. If we LLL-reduce a basis that is a solution of the $K$-DSP, it will still be a solution to the problem. Therefore, there exists an LLL-reduced basis that is the solution of the $K$-DSP. 

A solution output basis can be denoted as $\mathbf{{B}}_{out}$ as
\begin{equation}
  \begin{bmatrix}
   \mathbf{v}_{1} \\
   \vdots \\
   \mathbf{v}_{K}
 \end{bmatrix}
 = 
  \begin{bmatrix}
   x_{11} & x_{12} & \hdots & x_{1N}\\           
    x_{21} & x_{22} & \hdots & x_{2N} \\
    \vdots & \vdots & \ddots & \vdots\\
    x_{K1} & x_{K2} & \hdots & x_{KN} 
  \end{bmatrix}
  \begin{bmatrix}
   \mathbf{b}_{1} \\
   \vdots \\
   \mathbf{b}_{N}
 \end{bmatrix}.
 \label{eq: relation}
\end{equation}
The linear system in Eq.~(\ref{eq: relation}) can also be expressed in terms of matrices as $\mathbf{{B}}_{out} = \mathbf{X}\mathbf{B}_{P} $, where $\mathbf{{B}}_{out}$ and $\mathbf{X}$ are $K \times N$ matrices, and $\mathbf{B}_{P}$ is an $N\times N$ matrix, considering the worst-case assumption. Then, we have the inequality relating the infinity norm and the spectral matrix norm 
\begin{equation}
    ||\mathbf{X}||_{\infty} \leq ||\mathbf{{B}}_{out}||_{sp}||\mathbf{B}_{P}^{-1}||_{sp},
\end{equation}
Hence, using the bound on the basis $\mathbf{B}_{P}$ and $\mathbf{{B}}_{out}$ that do not have a gap, and the inequality obtained in Lemma~\ref{lemma:unimodular_bound}, we can write  
\begin{eqnarray}
    ||\mathbf{X}||_{\infty} \leq N\left (\frac{4}{3} + \varepsilon \right)^{3(N-1)/4}.
    \label{eq:bounds_coeff}
\end{eqnarray}

The coefficients that describe the vectors that span the densest sub-lattice must be bounded by $-2^{m} \leq x_i \leq 2^{m} -1$ for all $i =1,..., N$, where $m$ are the qubits that suffice to represent the coefficients $x_i$. 
Consequently, the total number of qubits can be expressed as 
\begin{equation}
    n = K \sum_{i = 1}^N m =K \sum_{i = 1}^N \lfloor \log{2^{m}} \rfloor \leq K\log{\prod_{i=1}^N 2^{m}}.
    \label{eq:total_n}
\end{equation}
Using Eq.~(\ref{eq:bounds_coeff}) we can write
\begin{align}
    \prod_{i=1}^N 2^{m} \leq \prod_{i=1}^N  N\left (\frac{4}{3} + \varepsilon \right)^{3(N-1)/4} \\ = N^N\left (\frac{4}{3} + \varepsilon \right)^{3(N^2-N)/4}.
    \label{eq:prod_ineq}
\end{align}
Taking the logarithm of Eq.~(\ref{eq:prod_ineq}) and substituting into Eq.~(\ref{eq:total_n}), we obtain that $\frac{3KN^2}{4}\log{\left(\frac{4}{3} + \varepsilon \right)}-\frac{3KN}{4}\log{\left(\frac{4}{3} + \varepsilon \right)} + N \log{N}$ qubits suffices to find the solution of the $K$-DSP using the quantum solver. 
\end{proof}
We can calculate the run-time in the context of Groverization and consider LLL preprocessing. However, the bound obtained can be reduced by allowing more substantial reductions in the input basis. Therefore, we present the Grover speedup for the latter case.

The pseudocode of the full algorithm for the $K$-DSP is provided in Algorithm 1, which incorporates the preprocessing step, involving LLL reduction, and the QAOA steps. 

\begin{algorithm}[H]
\caption{A quantum algorithm for the $K$-DSP\label{alg:algorithm}}
\begin{algorithmic}[1]
\Require $\mathbf{B}_{in}\subseteq \mathbb{R}^N$, $K \colon 0<K< \textnormal{dim}(\mathbf{B}_{in})$
\Ensure  $\mathbf{B}_{out} \subseteq \mathbb{R}^K$ a linear subspace.
\State $\mathbf{B}_{L} \gets$ Run LLL on $\mathbf{B}_{in}$
\If{$\mathbf{B}_{L}$ is gap-free}
\State  $\mathbf{B}_{P} \gets \mathbf{B}_{L}$
\ElsIf{$\mathbf{B}_{L}$ has gaps}
\State By Lemma~\ref{lem:gap-reduction}, reduce the dimension of the $K$-DSP 
\State  $\mathbf{B}_{P} \gets \mathbf{B}_{L} \subseteq \mathbb{R}^{(N-p)}$
\EndIf

\State $n \gets KN^2$
\State $\mathbf{G} \gets \mathbf{B}_P\mathbf{B}_P^T$  

\Procedure{MakeQaoa}{$ H_M,H_{DSP}, \boldsymbol{\theta}\colon \mathrm{array}, p \colon \mathrm{int}$}
\State $|\psi\rangle \gets H^{\otimes n}|0 \rangle^{\otimes n}$
\For{$t = 1$ to $p$}
    \State $|\psi\rangle \gets H_M H_{DSP} |\psi\rangle$
\EndFor
\State \textbf{return} $|\psi\rangle$
\EndProcedure

\Procedure{TrainQaoa}{$|\psi\rangle\colon \mathrm{array}$}
\State $\textnormal{old} = 0$
\While{True}
\State $|\psi\rangle \gets\Call{MakeQaoa}{ \boldsymbol{\theta}}$
\State $\boldsymbol{\theta}_{opt} \gets $ Compute gradient of $C(\boldsymbol{\theta})$ 
\State $\textnormal{error} \gets \langle\psi| C(\boldsymbol{\theta}_{opt})|\psi\rangle$
\If{$|\textnormal{error} - \textnormal{old}| < \textnormal{tol}$}
\State \textbf{break} 
\EndIf
\State $ \textnormal{old} = \textnormal{error}$
\EndWhile
\State \textbf{return} $\boldsymbol{\theta}_{opt}$
\EndProcedure
\State $\boldsymbol{\theta}_{opt}\gets\Call{TrainQaoa}{|\psi\rangle}$
\State $|\psi\rangle \gets\Call{MakeQaoa}{ H_M,H_{DSP},\boldsymbol{\theta}_{opt},p}$
\State $\mathbf{B}_{out} \gets$ Measure $|\psi\rangle$ in the $Z$ basis and post-processing
\State \textbf{return} $\mathbf{B}_{out}$
\end{algorithmic}
\end{algorithm}

\subsubsection{Preprocessing with an SVP oracle}
If we are allowed to use an SVP oracle, we can detect if there exists a basis with a gap in the lattice: it suffices to run dual-HKZ followed by HKZ~\cite{korkine1877formes}.  These alternative algorithms enable more powerful reductions of the input basis, although they may require certain assumptions. Then, only the gap-free case subsists in our Theorem, and it is possible to decrease the $(4/3 +\varepsilon)^{(N-1)/4}$ term that arises from the LLL-bound in Theorem~\ref{theorem:spatial}. In particular, the number of qubits becomes $O(5KN\log{N})$, which is an alternative construction of Dadush-Micciancio Lemma (see App.~\ref{app:appendixB} for the proof).

\begin{theorem}[Runtime of Groverized Exhaustive Search]
Let $N^{5KN}$ be the size of the search space, and let $M$ be the number of solutions in the space. Then, the runtime of Groverized exhaustive search for finding a solution of the $K$-DSP is 
\begin{equation*}
 O\left(\frac{N^{5KN/2}}{\sqrt{M}}\right).   
\end{equation*}
\label{theorem:kdsp-groverization}
\end{theorem}

\begin{proof}
    Consider the scenario where more powerful reduction algorithms are permitted in the classical preprocessing step. Then, the number of qubits required to ensure that at least one solution of the $K$-DSP is found within the search space is $O(5KN\log{N})$, where $N$ and $K$ are the ambient lattice and the densest sub-lattice dimensions, respectively. Consequently, the database consists of $2^{5KN\log{N}}$ elements, which can also be expressed as $N^{5KN}$. Drawing upon Grover's algorithm for unstructured search~\cite{grover1996fast}, only $O\left(N^{5KN/2}M^{-1/2}\right)$ queries are necessary to maximize the probability of obtaining the target state, where $M$ is the number of solutions.
\end{proof}

In Theorem~\ref{theorem:kdsp-groverization}, we establish that the runtime for Groverizing the exhaustive search is bounded by $O\left(N^{5KN/2}M^{-1/2}\right)$. For the best-known classical algorithm, the runtime is $O(K^{KN})$. When comparing both, we observe that the runtime for the classical algorithm scales as $2^{KN\log{K}}$, while the quantum algorithm has slightly worse performance due to an additional logarithmic term in the exponent, resulting in $2^{5KN\log{N}/2-1/2\log{M}}$. The observation that our quantum solver appears moderately slower than the classical one could be attributed to certain aspects of the proof. It is worth noting that, in certain scenarios, the Half-Volume Problem is suspected to be easier than the Shortest Vector Problem (SVP), particularly in the context of overstretched NTRU lattices~\cite{hoffstein1998ntru,fouque2018falcon}. Therefore, one possible explanation for the presence of this extra term may be linked to the slower preprocessing associated with the utilization of the SVP oracle.

A similar run-time guarantee as in Theorem~\ref{theorem:kdsp-groverization} is obtained when the first excited states of the Hamiltonian are found using quantum Gibbs sampling~\cite{chifang2023quantum}. 

\subsection{Ground state penalization}\label{subsec: gs_penalization}

As has been previously mentioned, the Hilbert space of the problem contains trivial solutions such as the zero vector or sub-lattices spanned by linearly dependent vectors. In this section, we propose a method to penalize the energy of the trivial solutions, although many other approaches could be considered~\cite{Higgott_2019}. The easiest way is to merely project into the sub-space of non-trivial solutions since the sub-space of trivial solutions has a simple algebraic characterization of having determinant zero. But this destroys the locality of the Hamiltonian. There is an incentive to keep the locality, either for ease of implementation or for run-time guarantees such as the one offered by quantum Gibbs sampling~\cite{chifang2023quantum,chifang2023sparse}.  

Let the problem Hamiltonian for the 2-Densest Sub-lattice Problem be the one presented in Eq.~(\ref{eq:ham_pen}). This Hamiltonian is a sum of up to 4-body hermitian matrices. Its eigenvectors are the different sub-lattices that can be generated with the available qubits of the problem, and the eigenvalues are the squared covolumes of respective sub-lattices. Thus, the ground-state space of $H_{DSP}$, which corresponds to the eigenvalue 0, is composed of trivial solutions such as the sub-lattice spanned by zero vectors and sub-lattices spanned by $K$-linearly dependent vectors.

The main goal is to find the first excited state of $H_{DSP}$. Therefore, we need to penalize somehow the ground state of the Hamiltonian. To achieve this, we can add an extra term to the Hamiltonian, similarly to the Projection Lemma in~\cite{kempe2005complexity}. In this context, we can write 
\begin{equation}
 H = H_{DSP} + r e^{-s H_{DSP}}, 
 \label{eq:ham_pen_2}
\end{equation}
where $r e^{-s H_{DSP}}$ corresponds to the penalizing term, and $r$ and $s$ are positive constants to be determined.

The extra term has vanishingly small eigenvalues for excited configurations, while its eigenvalue is equal to $1$ when considering trivial solutions. Therefore, the eigenvalues of the total Hamiltonian $H$ for the excited state space of $H_{DSP}$ are nearly the same as the eigenvalues of $H_{DSP}$. However, the eigenvalues of $H$ for the ground-state space of $H_{DSP}$ are now $r$. Since $r$ and $s$ are parameters to be tuned, we can associate them with a value such that the energy of the first excited state of $H_{DSP}$ is lower than the energy of trivial configurations. 

It is important to note several things here. 
Firstly, since the problem Hamiltonian is Hermitian, it is also the case for the extra term and thus, for the total Hamiltonian. 
Secondly, both terms of $H$ share the same eigenvectors, so $H$ has the same eigenvectors as the $H_{DSP}$. Finally, the addition of the second term implies that the Hamiltonian in Eq.~(\ref{eq:ham_pen_2}) does not need to be 4-body. 

The Hamiltonian proposed in Eq.~(\ref{eq:ham_pen_2}) has such a complex shape that cannot be implemented using QAOA. Therefore, for the purpose of simulation, a second-degree approximation could be considered such that 
\begin{equation}
    H \approx r \mathds{1} + (1-r s) H_{DSP} + r \frac{s^2}{2}  H_{DSP} H_{DSP}.
\end{equation}
This approximation is equivalent to considering $ H = (H_{DSP} - E)^2$ with some $E$ that depends on $r$ and $s$, together with an overall shift and re-scaling. Thus, any of the formulations is feasible for the QAOA emulation.

\subsection{Spectral gap bounding}

The optimal values for $r$ and $s$, or an approximate value for $E$ need to be determined to achieve such penalization. A tight lower bound on the spectral gap would be enough to either tune the parameters or specify $E$. However, proving lower bounds on the spectral gap is considered a hard problem~\cite{gottesman2010quantum,ambainis2014physical}. To solve this issue, we can devise our algorithm with a promised spectral gap $\epsilon_{DSP}$ as a parameter. The true spectral gap can be readily estimated using binary search, by using the parameterized algorithm as an oracle. 
Here, the first excited state energy, which is equivalent to the covolume squared of the densest sub-lattice, corresponds to the spectral gap. Therefore, the upper limit of the binary search algorithm can be set by bounding the Gramian formed by the set of vectors that span the densest sub-lattice.

\begin{lemma}\textnormal{(Upper bound on the spectral gap of $H_{DSP}$).}
Let the input of the algorithm be an $N$-dimensional LLL-reduced basis that defines a lattice $\mathcal{L}$. Then, the spectral gap of the problem Hamiltonian $H_{DSP}$ is bounded by $\Delta E_{DSP} \leq \left(\frac{4}{3} + \varepsilon \right )^{K(N-1)} \textnormal{vol}(\mathcal{L})^{2K}$.    \label{lemma:spectralgap} 
\end{lemma}
\begin{proof}
From Rankin's constant definition in Eq.~(\ref{eq:rankin}), 
\begin{equation}
    \textnormal{det}(\mathbf{G}(\mathbf{v}_1,...,\mathbf{v}_K))\leq \gamma_{N,K}  \textnormal{vol}(\mathcal{L})^{2K/N},
    \label{eq:bound_det}
\end{equation}
since $ \textnormal{det}(\mathbf{G}(\mathbf{v}_1,...,\mathbf{v}_K)) = \textnormal{vol}(\hat{\mathcal{L}})^{2}$. Rankin's constants are upper bounded by 
\begin{equation}
    \gamma_{N,K} \leq \left( \frac{\prod_{i=1}^K||\mathbf{v}_i||}{\textnormal{vol}(\mathcal{L})^{K/N}}\right)^2,
    \label{eq:gamma_bound}
\end{equation}
as shown in~\cite{nguyen2009hermite}. As previously mentioned, if we LLL-reduce a basis that is a solution of the $K$-DSP, it remains a valid solution to the problem. Therefore, there exists an LLL-reduced basis that is the solution of the $K$-DSP. Hence, the length of the vectors $||\mathbf{v}_i||$ that span the solution of the problem is bounded by Corollary~\ref{cor:relative-lll-basis-size}. Therefore,
\begin{equation}
    \textnormal{det}(\mathbf{G}(\mathbf{v}_1,...,\mathbf{v}_K))\leq \left (\frac{4}{3} + \varepsilon \right )^{K(N-1)} \textnormal{vol}(\mathcal{L})^{2K}.
\end{equation}
\end{proof}

\section{\label{sec:exp}Experimental results}
In this section, we present the results obtained after running a quantum emulation of the QAOA algorithm on a classical computer to discuss the performance of the $K$-DSP quantum solver for $K =2$. While these results are low-dimensional, they are illustrative of what we can obtain in higher dimensions. The results are presented for $N = 3$ and $N = 4$, as a function of $p$ and the quality of the bases. Note that, relevant lattices in cryptography have dimensions up to $400$. However, the rank of the input lattices has been limited to $4$, since the number of qubits scales as $O(KN^2)$, as shown in Theorem~\ref{theorem:spatial}. Thus, we fix the number of qubits per qudit to $2$.  

\begin{figure*}[t!]
  \begin{subfigure}[b]{2\columnwidth}
    \includegraphics[width=\columnwidth]{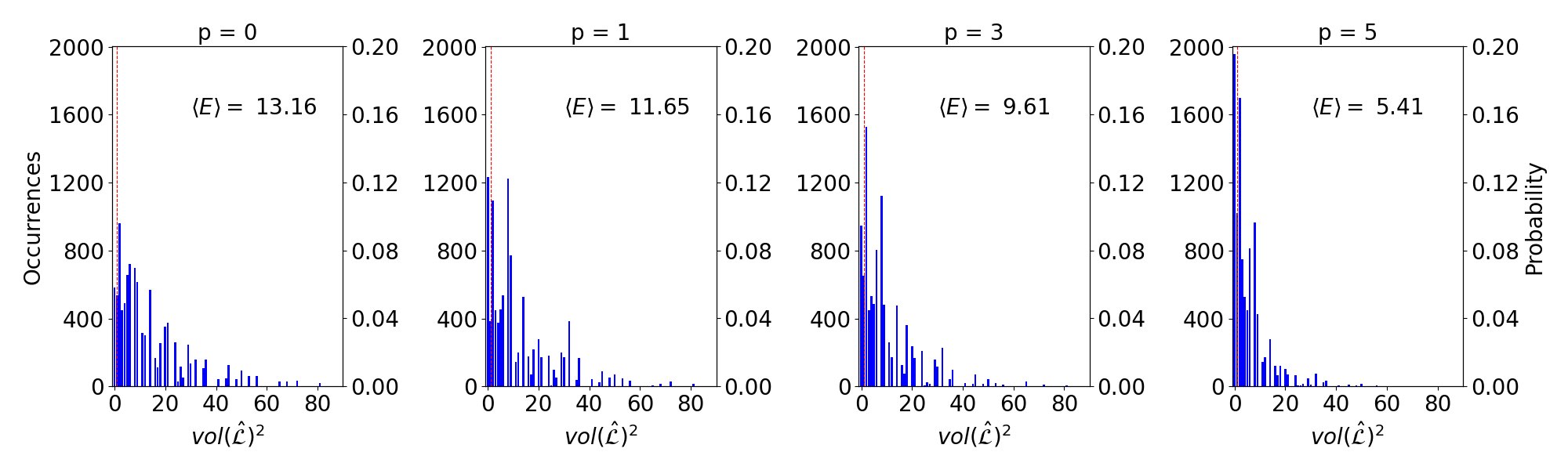}
    \caption{Good basis as input}
    \label{fig:3D_good_basis}
  \end{subfigure}
  \vspace{0.01mm}
  \begin{subfigure}[b]{2\columnwidth}
    \includegraphics[width=\columnwidth]{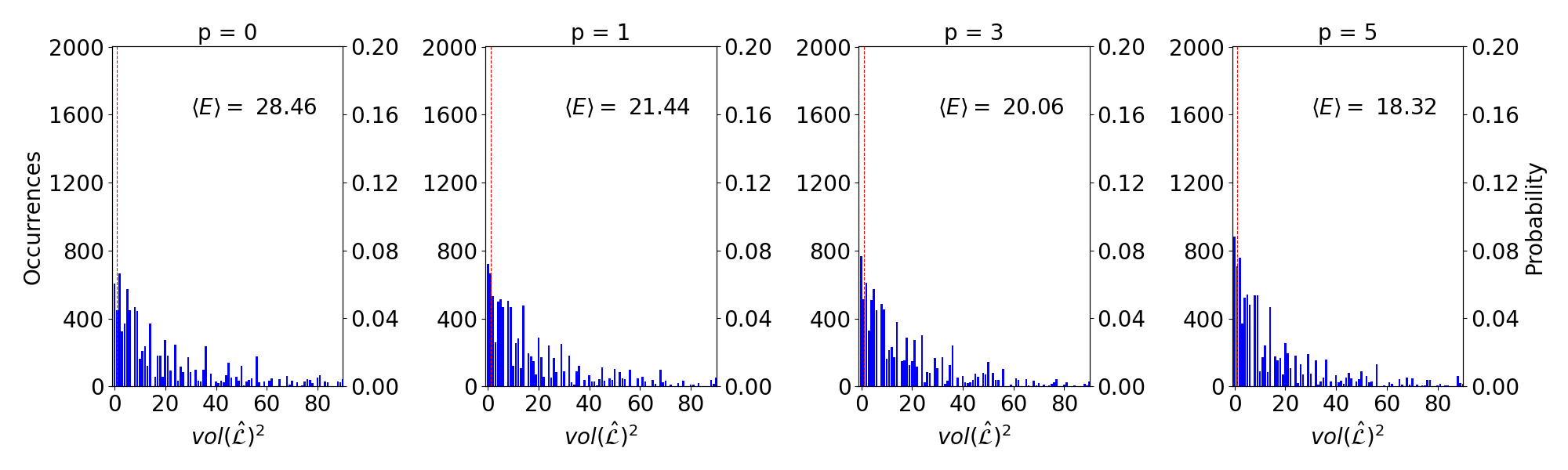}
    \caption{Worse basis as input}
    \label{fig:3D_worse_basis}
  \end{subfigure}
  \caption{The number of occurrences (left y-axis) and probability (right y-axis) represented by the blue bars for each eigenvalue of the Hamiltonian in Eq.~(\ref{eq:ham_pen}) for a 3D lattice. The number of layers increases from left to right with values set to $p = 0,1,3,5$. The figures include the average energy calculated from $10,000$ samples. The red dashed line points out the location of the solution of the $K$-DSP equivalent to the densest sub-lattice. (a) uses the 3D good basis as input, while (b) uses a worse basis achieved by multiplying the good basis by a unimodular matrix.}
  \label{fig:3D_lattice}
\end{figure*}

\begin{figure*}[t!]
  \begin{subfigure}{2\columnwidth}
    \includegraphics[width=\columnwidth]{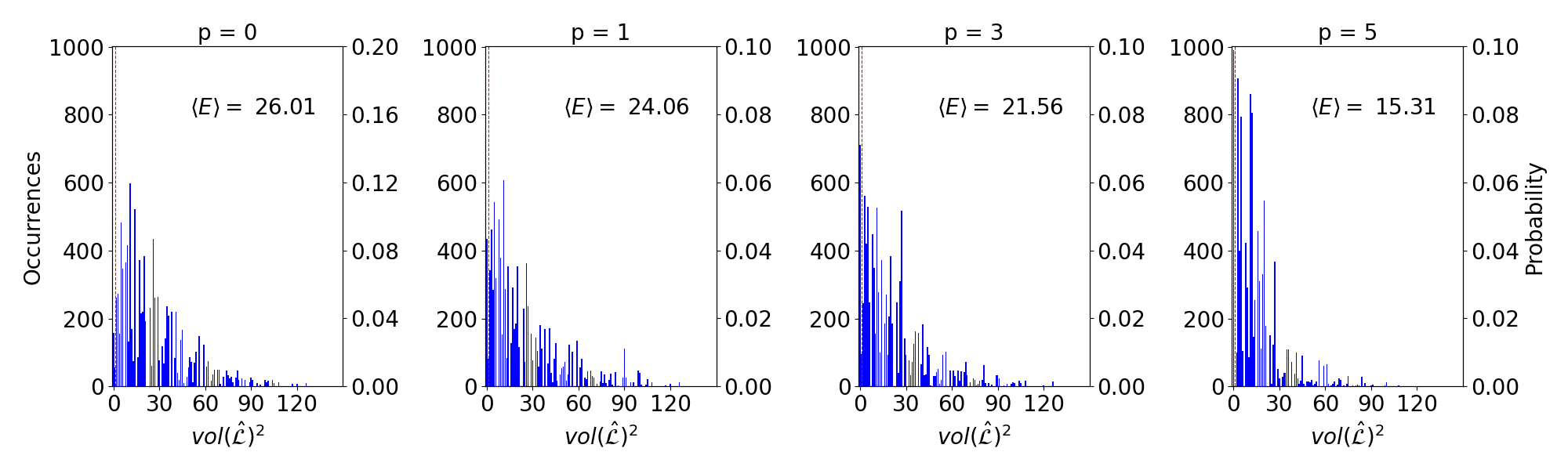}
    \caption{Good basis as input}
    \label{fig:4D_good_basis}
  \end{subfigure}
  \vspace{0.01mm}
  \begin{subfigure}{2\columnwidth}
    \includegraphics[width=\columnwidth]{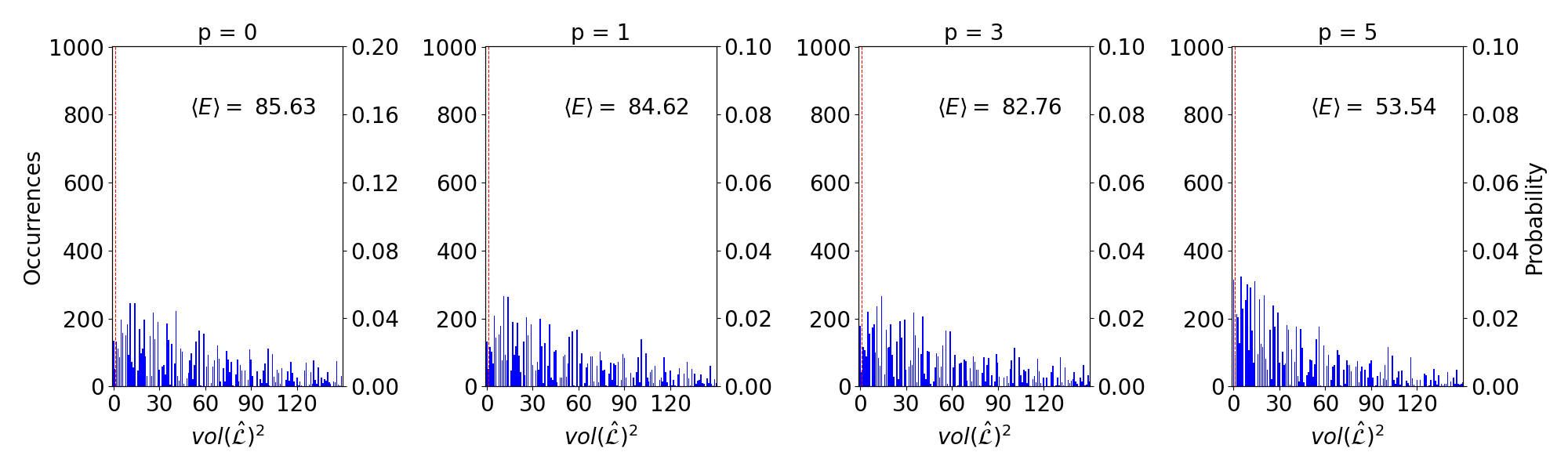}
    \caption{Worse basis as input}
    \label{fig:4D_worse_basis}
  \end{subfigure}
  \caption{The number of occurrences (left y-axis) and probability (right y-axis) represented by the blue bars for each eigenvalue of the Hamiltonian in Eq.~(\ref{eq:ham_pen}) for a 4D lattice. The number of layers increases from left to right with values set to $p = 0,1,3,5$. The figures include the average energy calculated from $10,000$ samples. The red dashed line points out the location of the solution of the $K$-DSP equivalent to the densest sub-lattice. (a) uses the 4D good basis as input, while (b) uses a worse basis achieved by multiplying the good basis by a unimodular matrix.}
  \label{fig:results4}
\end{figure*}

Theorem~\ref{theorem:spatial} and Lemma~\ref{lemma:spectralgap} have been proven for LLL-reduced input bases. Nevertheless, for experiments in low-dimensional lattices (therefore classically simulable), we create input bases in a different manner, as for low dimensions LLL returns trivially short bases which are not illustrative of the capabilities of the quantum $K$-DSP algorithm.
Instead, we first generate short bases, then we scramble them by multiplying them with random unimodular matrices. In essence, we make the short basis slightly worse for the simulations. The effect is that the optimal solutions we seek from the algorithm are not trivial ones. By scrambling the basis of the lattices, we ensure that the complexity of finding short vectors is similar to that of random lattices, based on the hardness assumption of lattice isomorphism~\cite{cryptoeprint:2021/1332}. This assumption implies that randomly rotated and scrambled bases are indistinguishable from random lattices, suggesting our experiments are reflective of the true difficulty of the problem. 

Although we evaluate the spatial scaling performance of the QAOA for the $2$-DSP, it is not possible to extrapolate the success probability heuristic to cryptographically relevant dimensions, though applying Grover search to the search space does give some complexity scaling.

In our implementation of QAOA, we use TensorFlow Quantum along with the Keras API for parameter optimization. The process begins by initializing the circuit parameters ($\boldsymbol{\gamma}, \boldsymbol{\beta}$) to values uniformly sampled from the range of values $[0, 2\pi]$. These parameters are treated as trainable variables in Keras. We employ the Adam optimizer, a widely-used gradient-based optimization method, with a learning rate of 0.001 to iteratively minimize the cost function, and a noiseless backend throughout the training process. Gradients of the cost function with respect to the parameters are computed using TensorFlow Quantum's differentiator, which efficiently calculates the gradient. Training continues until either the change in the cost function value falls below a small tolerance level (1e-6) or a maximum of 1000 epochs is reached. It is important to note that, due to the highly non-convex landscape typical of QAOA problems, the circuit parameters obtained might not be the optimal parameters. Rather, the training aims to converge to near-optimal values that satisfy the specified stopping criteria. Note that, although we used default values to initialize the circuit parameters, several alternative methods could have been employed. These include heuristic approaches~\cite{PhysRevX.10.021067}, parameter fixing techniques~\cite{9605323}, leveraging Graph Neural Networks (GNNs)~\cite{Jain2022graphneuralnetwork}, sharing and reusing parameters across different QAOA instances~\cite{9605328, Shaydulin_2019, 9651381}, and strategies inspired by Trotterized Quantum Annealing~\cite{Sack2021quantumannealing}. In the same context, other gradient-based optimization techniques could have been implemented~\cite{crooks2018performancequantumapproximateoptimization, Streif_2020, Sweke2020stochasticgradient, Sung_2020, yao2020policygradientbasedquantum, Lotshaw2021}, as well as gradient-free approaches~\cite{ACAMPORA2023110296, Cheng2024, DONG2020242}.

In Fig.~\ref{fig:3D_lattice} we present the histograms obtained after training our quantum algorithm for $N = 3$ and for different values of $p$. First, the QAOA was trained using the expectation value of the energy as the cost function. Then, we drew $10,000$ samples, which were measured in the $Z$ basis. The figures in Fig.~\ref{fig:3D_good_basis} show the results when using a $3$-dimensional short basis that defines the integer lattice. In contrast, the plots in Fig.~\ref{fig:3D_worse_basis} exhibit the outcomes when the short basis was multiplied by a unimodular matrix to obtain a lower quality basis (long basis) for the same lattice as input. 

Recall that, classically, it is more challenging to find the solutions for the $K$-DSP problem when given bad input bases. The difficulty of the QAOA in finding ground-state solutions increases with the number of qubits due to the exponentially increasing search space. Nevertheless, as the formulation for the ground state penalization has not been implemented, the outputs contain trivial solutions (i.e. the zero vector or linearly dependent vectors) that correspond to the lowest energy states.

In both the top and bottom sets of histograms in Fig.~\ref{fig:3D_lattice}, the $x$-axis depicts the covolumes squared of the sub-lattices, which are equivalent to the eigenvalues of the Hamiltonian in Eq.~(\ref{eq:ham_pen}). The blue bars represent the number of occurrences of the covolumes squared after 10,000 samples in the left $y$-axis, and the equivalent probability (occurrences divided by 10,000) in the right-hand $y$-axis. The thin red dashed line shows the location of the state that corresponds to the densest sub-lattice. The $x$-axis has been truncated so a small number of very high-energy results are not shown in order to improve the readability of the lower-energy samples.
The number of layers increases from left to right and the left-hand subplot in each row ($p = 0$) corresponds to uniform random sampling from the search space, thus representing a benchmark against which the nonzero $p$ can be compared.

Analyzing the results, we can observe that in both sets the average energy of the system decreases with $p$ as expected (from 13.16 at $p=0$ to 5.41 at $p=5$ for the canonical basis, and from 28.46 to 18.32 for the worse basis). In Fig.~\ref{fig:3D_good_basis}, the average expected value of the energy has been reduced by $60\%$, while in Fig.~\ref{fig:3D_worse_basis} only by $35\%$. This shows that the algorithm's performance is notably enhanced when using better bases. This underscores the importance of classical preprocessing through LLL-reduction in cases where the algorithm is executed on more challenging lattices. 
The trend is also visible on the blue bars since the probability mass is more concentrated on the left-hand side of the subplots for the higher $p$ values. In this way, the ground state reaches the highest probability among all possible solutions for $p=5$ in both the top and bottom subplots. At $p = 5$, the histograms exhibit a distribution similar to the Gibbs distribution~\cite{Li_2020}, where the number of occurrences decreases with the energy of the different eigenvalues.  

Notice that, while for good bases almost all the data is concentrated in the first 10 excited states, for bad bases the occurrences are more evenly distributed across the different bins, which involves a reduction of the number of occurrences with respect to the top row figures for each state. Therefore, it is easier for the algorithm to find low-energy solutions when the input consists of short and close to orthogonal vectors, as it clearly presents a better performance in this situation. Nevertheless, the ideal QAOA output would consist of obtaining the ground state with probability $1$ (the adiabatic limit), represented by a blue bar of height $10,000$ at the $0$ value on the $x$-axis. This would also be the ideal scenario for $K$-DSP when ground-state penalization is implemented.

In the same way as for $N = 3$, in Fig.~\ref{fig:results4}, we present the histograms obtained when using as input a $4$-dimensional good basis in the right-hand-side figure, and a bad basis in the left-hand-side figure. The $x$-axis has also been truncated to enhance the legibility of the data.

The behavior is quite similar to the one in Fig.~\ref{fig:3D_lattice}: the average energy of the system decreases with $p$ in both sets, being lower for the set of good bases (Fig.~\ref{fig:4D_good_basis}) than for bad bases (Fig.~\ref{fig:4D_worse_basis}) at each $p$. The occurrences at high-energy states decrease as $p$ increases, while low-energy configurations become more common. The taller blue bars on the left side of the subplots at higher $p$ indicate higher probability density in this region. As was observed in Fig.~\ref{fig:3D_lattice}, the outputs are less concentrated around low-energy values for the bad bases. 

Nevertheless, noticeable differences arise when considering different dimensions of the input basis. Firstly, since solving the $2$-DSP using a $4$-dimensional basis requires a larger search space, it results in a more complex QAOA circuit. The results for $N = 4$ have been obtained using $16$ qubits, compared to $12$ qubits for $N = 3$. In App.~\ref{app:appendixC}, we can observe that the QAOA circuit is relatively complex even for a 6-qubit system. For 4-dimensional good bases, the improvement in getting low-energy configurations when increasing $p$ is more modest than for 3-dimensional good bases, while for 4-dimensional bad bases is barely discernible.

We can zoom in on the results by analyzing the probability for low-energy configurations and comparing their behavior with increasing $p$. The probability of obtaining states in ranges $\textnormal{vol}^2(\mathcal{\hat{L}}) \leq 5, 10, 20$ is represented in Table~\ref{table:lowenergy}. We show the probability with respect to $p$ and the dimension of the input basis, as well as its quality.  

The table illustrates that increasing the number of layers from $p = 1$ to $p = 5$ in all cases enhances the probability of obtaining low-energy states. Notice that, when using $3$-dimensional good bases, we obtain the best performance of our quantum solver since the probability goes from 0.40 to 0.50 for $\textnormal{vol}^2(\mathcal{\hat{L}}) \leq 5$ and from 0.81 to 0.88 for $\textnormal{vol}^2(\mathcal{\hat{L}}) \leq 20$. This means that in the range $\textnormal{vol}^2(\mathcal{\hat{L}}) \leq 20$, the $81\%$-$88\%$ of the outcomes are found within this low-energy range versus the $10\% - 20\%$ obtained between $20 \leq \textnormal{vol}^2(\mathcal{\hat{L}}) \leq 261$. 

Still, for bad bases, the impact of increasing the number of layers on the probability of obtaining low-energy states is significantly lower. The worst case performance is found for a $4$-dimensional bad basis where only $26\% - 30\%$ are within the range $\textnormal{vol}^2(\mathcal{\hat{L}}) \leq 20$. This likely indicates that much higher depths are required to see the same improvements, which would imply that any classical LLL-type preprocessing could reduce the work and improve the performance of the quantum $K$-DSP algorithm.

For bad bases, the probability at each $p$ is approximately 0.1 lower compared to good bases. This behavior was expected because the average covolume of a sub-lattice when randomly sampled will be higher for bad input bases and also the integer values that represent the coefficients, both introducing additional costs. 

\vspace{5mm}
\begin{table}[ht]
\begin{tabular}{lcllll}
\toprule
 \multicolumn{1}{c}{{\textbf{Range}}} & \multicolumn{1}{c}{{ \textbf{Layers}}} & \multicolumn{2}{c}{$\mathbf{N = 3}$} & \multicolumn{2}{c}{$\mathbf{N = 4}$} \\
\cmidrule(lr){3-4} \cmidrule(lr){5-6}
\multicolumn{1}{c}{}& \multicolumn{1}{l}{} & good & bad & good & bad \\
\midrule
\multicolumn{1}{c}{}  & 1  & 0.399 &0.319 & 0.215 & 0.068 \\ 
\cmidrule(lr){2-2}
\multicolumn{1}{c}{$\textnormal{vol}(\hat{\mathcal{L}})^2 \leq 5$}& 3 & 0.429 & 0.325  & 0.236 & 0.072 \\ 
\cmidrule(lr){2-2}
\multicolumn{1}{c}{} & 5  & 0.499 & 0.342 & 0.264 & 0.090\\
\midrule
 & 1  &0.651 & 0.475 & 0.349  & 0.123 \\ 
 \cmidrule(lr){2-2}
\multicolumn{1}{c}{$\textnormal{vol}(\hat{\mathcal{L}})^2 \leq 10$} & 3  & 0.675  & 0.480  & 0.363 & 0.130\\ 
 \cmidrule(lr){2-2}
& 5 & 0.737 & 0.501 & 0.379  & 0.159 \\
\midrule
 & 1  & 0.812 & 0.676 & 0.595 & 0.257 \\ 
 \cmidrule(lr){2-2}
\multicolumn{1}{c}{$\textnormal{vol}(\hat{\mathcal{L}})^2 \leq 20$} & 3 & 0.841  & 0.673 & 0.606  & 0.260 \\ 
\cmidrule(lr){2-2}
& 5 & 0.881  & 0.689 & 0.668  & 0.302\\ 
\bottomrule
\end{tabular}
\caption{Probability of obtaining low-energy configurations within the ranges $\textnormal{vol}^2(\mathcal{\hat{L}}) \leq 5, 10, 20$ as a function of the number of layers $p$, the dimension $N$, and the quality of the input basis for 3D and 4D integer lattices.}
\label{table:lowenergy}
\end{table}

\subsection{Variational algorithm scaling}

Here, we discuss the complexity of our quantum algorithm in terms of the total number of 1-qubit and 2-qubit gates required for the QAOA implementation across various values of $N$. 
The results depicted in Fig.~\ref{fig:qaoa_scaling} were derived from calculating directly the number of gates for one single layer of QAOA ($p = 1$) for the $2$-DSP. This includes both 1-qubit gates (shown in blue) and 2-qubit gates (shown in red), with comparisons between good bases (solid lines) and bad bases (dashed lines).

Fig.~\ref{fig:qaoa_scaling} shows that the number of gates increases polynomially with $N$, but the growth rate is significantly higher when using a bad basis. In the good basis case, the number of 2-qubit gates grows steadily, and by $N=15$, the 2-qubit gates double the number of 1-qubit gates, indicating manageable complexity despite the growing lattice size.
In contrast, the bad basis case exhibits a much steeper increase in the number of gates, particularly the 2-qubit gates. For $N=15$, the number of 2-qubit gates is nearly six times greater than 1-qubit gates. This highlights the inefficiency introduced when using bad bases, as the high number of entangling gates significantly increases the complexity of the quantum circuit.
These observations highlight the importance of selecting a good basis to mitigate system complexity. While the number of 1-qubit gates remains comparable between cases, the significantly higher number of 2-qubit gates in the bad basis case makes the simulation far more resource-intensive. Note that, increasing the number of QAOA layers $p$ would simply scale the gate counts proportionally, as the circuit parameters change between layers but the structure remains the same.

Overall, these results reflect the widely known fact that training VQAs becomes expensive fast as the system complexity grows, which is indeed the result of using the bad basis of higher dimensions as input of the algorithm.

\begin{figure}
\includegraphics[width=1\columnwidth]{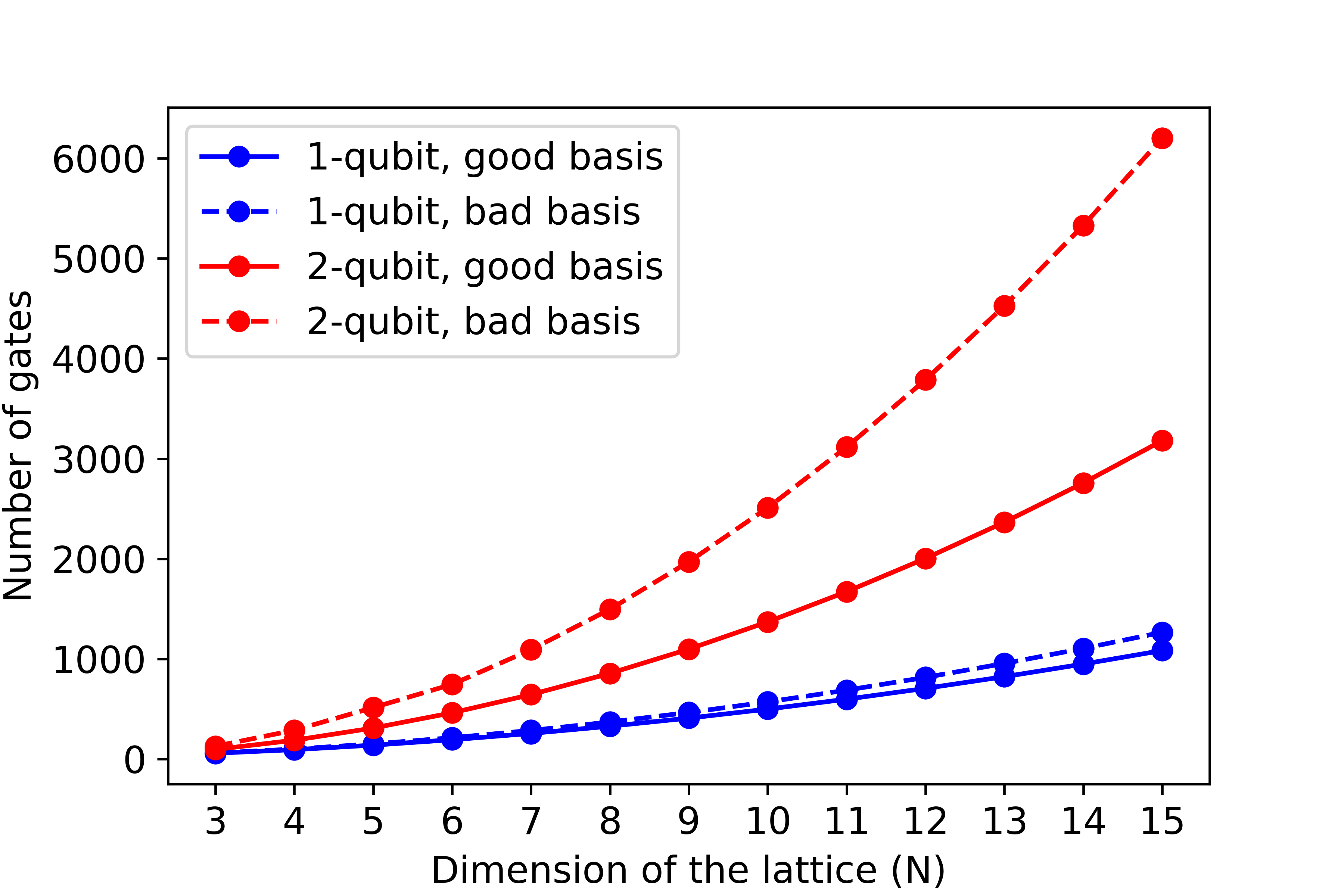}
\caption{Scaling behavior of one QAOA layer for the $2$-DSP. The total number of 1-qubit gates in the QAOA implementation is shown in blue, while the number of 2-qubit gates is depicted in red. Solid lines represent results for a good basis, and dashed lines indicate the findings for a bad basis, both plotted as a function of the lattice dimension $N$.}
\label{fig:qaoa_scaling}
\end{figure}

\section{\label{sec:conclusions}Conclusions}
While SVP is a widely studied problem in lattice literature and is closely linked to recently standardized cryptosystems such as Kyber, Dilithium, and Falcon~\cite{8406610, fouque2018falcon}, its generalization, $K$-DSP has remained much less tested. Due to the more general structure of $K$-DSP, further examination may enable deeper insights into the hardness of lattice problems. 

The work presented here serves as a platform to enable further exploration of this problem and gives some indication of the complexity underlying the structure of mathematical lattices. The formulation of the $K$-DSP Hamiltonian, which becomes extremely complex even for small systems as in App.~\ref{app:appendixC}, serves as a tool with which we can probe the problem from a quantum angle.

The bounding of the search space for the algorithm defined naturally leads to a Grover-based quantum time complexity, even if no such result exists for QAOA or AQC-based applications of the Hamiltonian defined here. We showed how to incorporate such a Groverized $K$-DSP search under certain conditions into a global hybrid algorithm, however, we believe that with stronger conditions on input bases, and reduced spatial requirements, future works could investigate faster hybrid quantum/classical algorithms for tackling $K$-DSP.

The experimental results presented demonstrate the ability of our quantum routine to sample dense lattices in low dimensions, but in order for future versions to be of practical use, many more tricks from state-of-the-art QAOA literature would need to be applied, such as improved cost functions like CVaR ~\cite{Barkoutsos_2020} or Gibbs objective functions~\cite{Li_2020}, and parameter initialization ~\cite{PhysRevX.10.021067, 9605323, Jain2022graphneuralnetwork, 9605328, Shaydulin_2019, 9651381, Sack2021quantumannealing}. Furthermore, for high $K$, the many-body $Z$ interactions present in the Hamiltonians will require expensive decompositions in order to implement on any future architectures ~\cite{boothby2020nextgeneration}, or architectures will have to adapt to be more accommodating to many-body interactions ~\cite{Jordan_2008,Chancellor_2017, leib2016transmon}. This `Hamiltonian engineering', whether directly of the $K$-DSP Hamiltonian defined here or via the construction of new Hamiltonians is thus another area for future work both from the quantum lattice algorithm, and the hardware design angles.

\begin{acknowledgments}
We thank Martin R. Albrecht for introducing us to this problem, and Miloš Prokop for the suggestions provided during the Hamiltonian formulation. We also thank Martin Ganahl and Rishi Sreedhar for discussions regarding the QAOA simulation, and Stefan Leichenauer for ideas on ground state penalization.
\end{acknowledgments}

\appendix
\section{\label{app:appendixA}Extended Hamiltonian formulation}

For the purpose of the QAOA simulation, the dimension of the sub-lattice that we wanted to find $K$ was fixed to 2. However, it can also be calculated the expression for $K = 3$. In this context, the Gramian (equivalent to the squared covolume of the sub-lattice) is spanned by the basis $\hat{\mathbf{B}} = (\mathbf{v}_1, \mathbf{v}_2, \mathbf{v}_3)$ can be written as 
\begin{equation}
\textnormal{det}(\mathbf{G}(\mathbf{v}_1,\mathbf{v}_2,\mathbf{v}_3)) = 
   \begin{vmatrix}
\langle\mathbf{v}_{1}, \mathbf{v}_{1} \rangle & \langle\mathbf{v}_{1}, \mathbf{v}_{2} \rangle &\langle\mathbf{v}_{1}, \mathbf{v}_{3} \rangle  \\
\langle\mathbf{v}_{2}, \mathbf{v}_{1} \rangle & \langle\mathbf{v}_{2}, \mathbf{v}_{2} \rangle &\langle\mathbf{v}_{2}, \mathbf{v}_{3}\rangle \\
\langle\mathbf{v}_{3}, \mathbf{v}_{1} \rangle & \langle\mathbf{v}_{3}, \mathbf{v}_{2} \rangle &\langle\mathbf{v}_{3}, \mathbf{v}_{3}\rangle
\end{vmatrix}. 
\label{eq:gram3}
\end{equation}

Considering that the three output vectors are linear combinations of the input basis, we can express them as 
\begin{equation}
    \begin{cases}
      \mathbf{v}_{1} = \mathbf{x} \mathbf{B} = x_1 \mathbf{b}_1 + ... +x_N\mathbf{b}_N,\\
      \mathbf{v}_{2} = \mathbf{y} \mathbf{B} = y_1 \mathbf{b}_1 + ... +y_N\mathbf{b}_N,\\
     \mathbf{v}_{3} = \mathbf{z} \mathbf{B} = z_1\mathbf{b}_1 + ... + z_N\mathbf{b}_N.
    \end{cases}   
    \label{eq:vectors3}
\end{equation}
Calculating the determinant of Eq.~(\ref{eq:gram3}) in terms of the vectors in Eq.~(\ref{eq:vectors3}), the cost function results in 
\begin{align}
    \textnormal{vol}(\mathcal{\hat{L}})^2 &= \sum_{i,j,k,l,m,n = 1}^N x_ix_jy_ky_lz_mz_n (\mathbf{G}_{ij}\mathbf{G}_{kl}\mathbf{G}_{mn}\nonumber \\
    & + \mathbf{G}_{ik}\mathbf{G}_{lm}\mathbf{G}_{nj}+\mathbf{G}_{im}\mathbf{G}_{kj}\mathbf{G}_{nl}- \mathbf{G}_{im}\mathbf{G}_{kl}\mathbf{G}_{nj}\nonumber \\
    & -\mathbf{G}_{ik}\mathbf{G}_{lj}\mathbf{G}_{mn}-\mathbf{G}_{ij}\mathbf{G}_{km}\mathbf{G}_{nl}).
\end{align}

To construct the classical Hamiltonian, we need to substitute each coefficient by the qudits operators $\hat{Q}^{(j)}$. 
Applying $H_{DSP}$ over a configuration of qubits, each operator will return the value of the corresponding coefficient and in this way, they will output the covolume squared of the sub-lattice which is an eigenvalue of the problem Hamiltonian and the sub-lattice, defined by the grid of qubits, as the eigenvector. 

\section{\label{app:appendixB}Bounds on the number of qubits for HKZ-reduced basis}
Considering that we have an extremely reduced input basis, the number of qubits that suffice to find an exact solution to the $K$-DSP can be reduced.

\begin{theorem}\textnormal{(Bound on the number of qubits).}\label{theorem:spatial2}
Let $N$-dimensional lattice $\mathcal{L}$ be the input to the quantum algorithm for the $K$-DSP as the span of a gap-free HKZ-reduced basis $\mathbf{B}= (\mathbf{b}_1,..., \mathbf{b}_N)$. Then, a $5KN\log{N}$ qubit Hilbert search space is sufficient to ensure that at least one exact solution of the $K$-DSP is contained. 
\end{theorem}

\begin{proof} In Theorem~\ref{theorem:spatial} is shown that the unimodular transformation $\mathbf{X}$ is bounded as 
\begin{equation}
    ||\mathbf{X}||_{\infty} \leq ||\mathbf{{B}}_{out}||_{sp}||\mathbf{B}_{in}^{-1}||_{sp},
    \label{eq:X_bound}
\end{equation}
\noindent
where $\mathbf{{B}}_{in}$ and $\mathbf{{B}}_{out}$ correspond to the input basis and the output basis, respectively.

Let the input be a gap-free HKZ-reduced basis~\cite{korkine1877formes}. Hence, the length of the vectors is bounded by 
\begin{eqnarray}
    |\mathbf{b}_i|^2 \leq \frac{i + 3}{4}\lambda_i(\mathcal{L})^2
    \label{eq:bi_bound_hkz}
\end{eqnarray}
When allowing for an SVP oracle for the preprocessing step, we ensure that the input and output basis are gap-free. Therefore, the successive minima $\lambda_i(\mathcal{L})$ of a lattice $\mathcal{L}$ are related as $\lambda_1(\mathcal{L}) \leq \lambda_2(\mathcal{L}) \leq \ldots \leq \lambda_N(\mathcal{L})\leq\sqrt{N}\lambda_1(\mathcal{L}) $. Taking the highest value for $i$ in Eq.~(\ref{eq:bi_bound_hkz}) to get an upper bound we obtain
\begin{eqnarray}
    |\mathbf{b}_i|^2 \leq \frac{N + 3}{4} N\lambda_1(\mathcal{L})^2.
    \label{eq:bi_bound_hkz2}
\end{eqnarray}
Using Minkowski's Theorem $\lambda_1(\mathcal{L}) \leq \sqrt{N}\textnormal{vol}(\mathcal{L})^{1/2N}$ the bound on the length of the basis vectors results in 
\begin{eqnarray}
    |\mathbf{b}_i|^2 \leq \frac{N + 3}{4} N^2\textnormal{vol}(\mathcal{L})^{1/N}.
    \label{eq:bi_bound_hkz3}
\end{eqnarray}
The result obtained in Eq.~(\ref{eq:bi_bound_hkz3}) can be substituted into Eq.~(\ref{eq:X_bound}), such that 
\begin{equation}
    ||\mathbf{X}||_{\infty} \leq  N\left(\frac{N + 3}{4}\right)^2 N^2 = \frac{N^5}{16}+ \frac{3N^4}{8} + \frac{9N^3}{16},
    \label{eq:X_bound2}
\end{equation}

The coefficients that describe the vectors that span the densest sub-lattice must be bounded by $-2^{m} \leq x_i \leq 2^{m} -1$ for all $i =1,...,N$, where $m$ are the qubits required to represent the coefficients $x_i$. 
Consequently, the total number of qubits required to run the $K$-DSP quantum solver can be expressed as 
\begin{equation}
    n = K  \sum_{i = 1}^N \lfloor \log{2^{m}} \rfloor \leq K  \log{\prod_{i=1}^N 2^{m}}.
    \label{eq:total_n2}
\end{equation}

Using the relation in Eq.~(\ref{eq:X_bound2}), we can write the inequality
\begin{align}
    \prod_{i=1}^N 2^{m} \leq& \prod_{i=1}^N \left (\frac{N^5}{16}+ \frac{3N^4}{8} + \frac{9N^3}{16}\right )  \nonumber\\
    & =
    \left (\frac{N^5}{16}+ \frac{3N^4}{8} + \frac{9N^3}{16}\right )^N  
\label{eq:ni2}  
\end{align}

Thus, the upper bounds for the number of qubits result in
\begin{equation}
    n \leq K \log{\left (\frac{N^5}{16}+ \frac{3N^4}{8} + \frac{9N^3}{16}\right )^N},
\end{equation}
which implies that the quantum algorithm for the $K$-DSP requires $O\left(5KN \log N\right)$ qubits. 
\end{proof}

\section{\label{app:appendixC}QAOA circuit diagram}

In Fig.~\ref{fig:QAOAcircuit} we show an illustration of a simple QAOA circuit generated using a \emph{Cirq} feature called \emph{SVGCircuit}, used to display quantum circuits. The circuit exemplifies the complexity of the problem, since a substantial amount of gates are already required to implement the QAOA for $K = 2$, $N = 3$, and only one qubit per qudit and one layer ($p = 1$). The problem Hamiltonian employed for the QAOA simulation is similar but more complex than the one presented in Eq.~(\ref{eq: cost}).

The first logic gate applied to all the qubits in Fig.~\ref{fig:QAOAcircuit} is the Hadamard gate, which generates a superposition of all possible states that can be created with the available qubits. The rest of the quantum logic gates represent the application of the problem operator and the mixer operator over the superposition state. The last line of the circuit represents the measurements that are performed at the end, which will be fed to the classical optimization process in which the QAOA parameters are updated.

\begin{figure*}[h]
   \begin{minipage}{0.92\textwidth}
    \includegraphics[width=\linewidth]{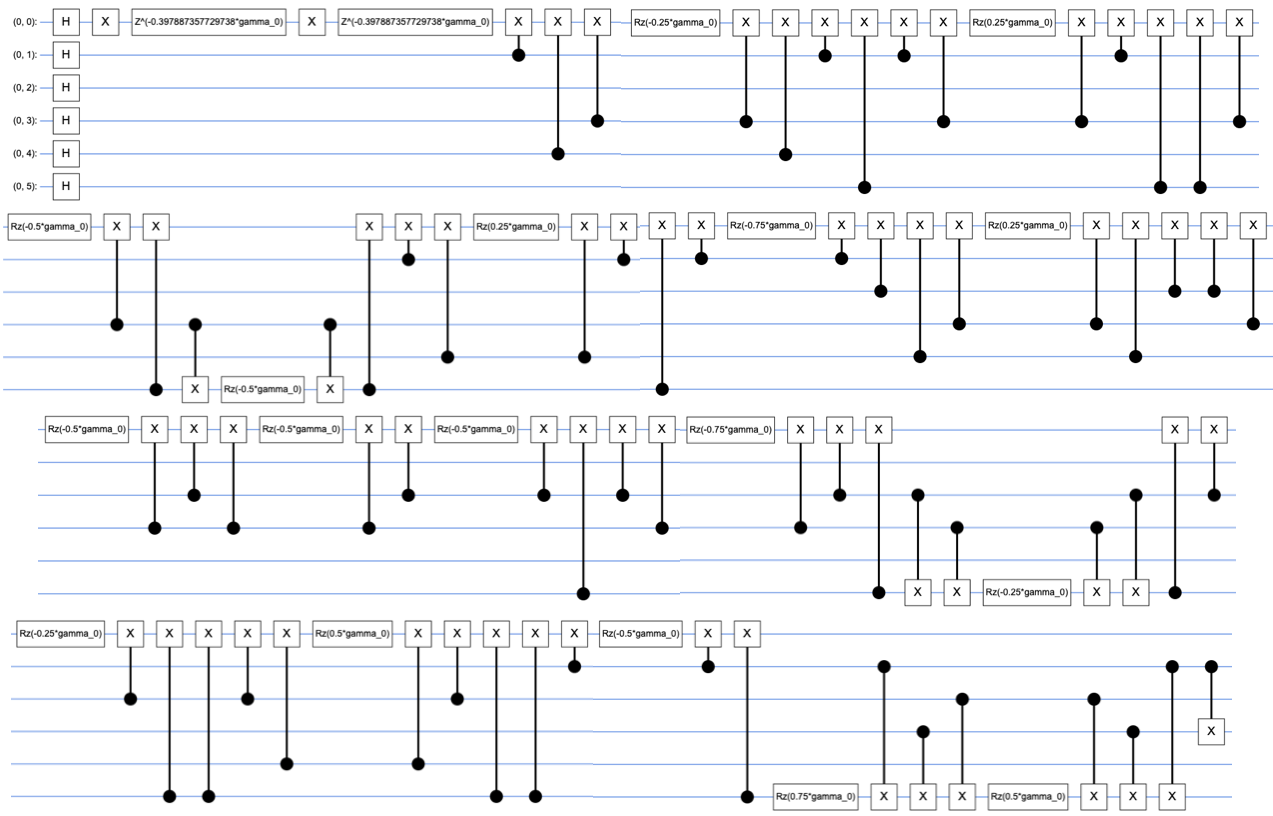}
    \label{fig:QAOA_diagram1}
   \end{minipage}
    \quad
   \begin{minipage}{0.98\textwidth}
    \includegraphics[width=\linewidth]{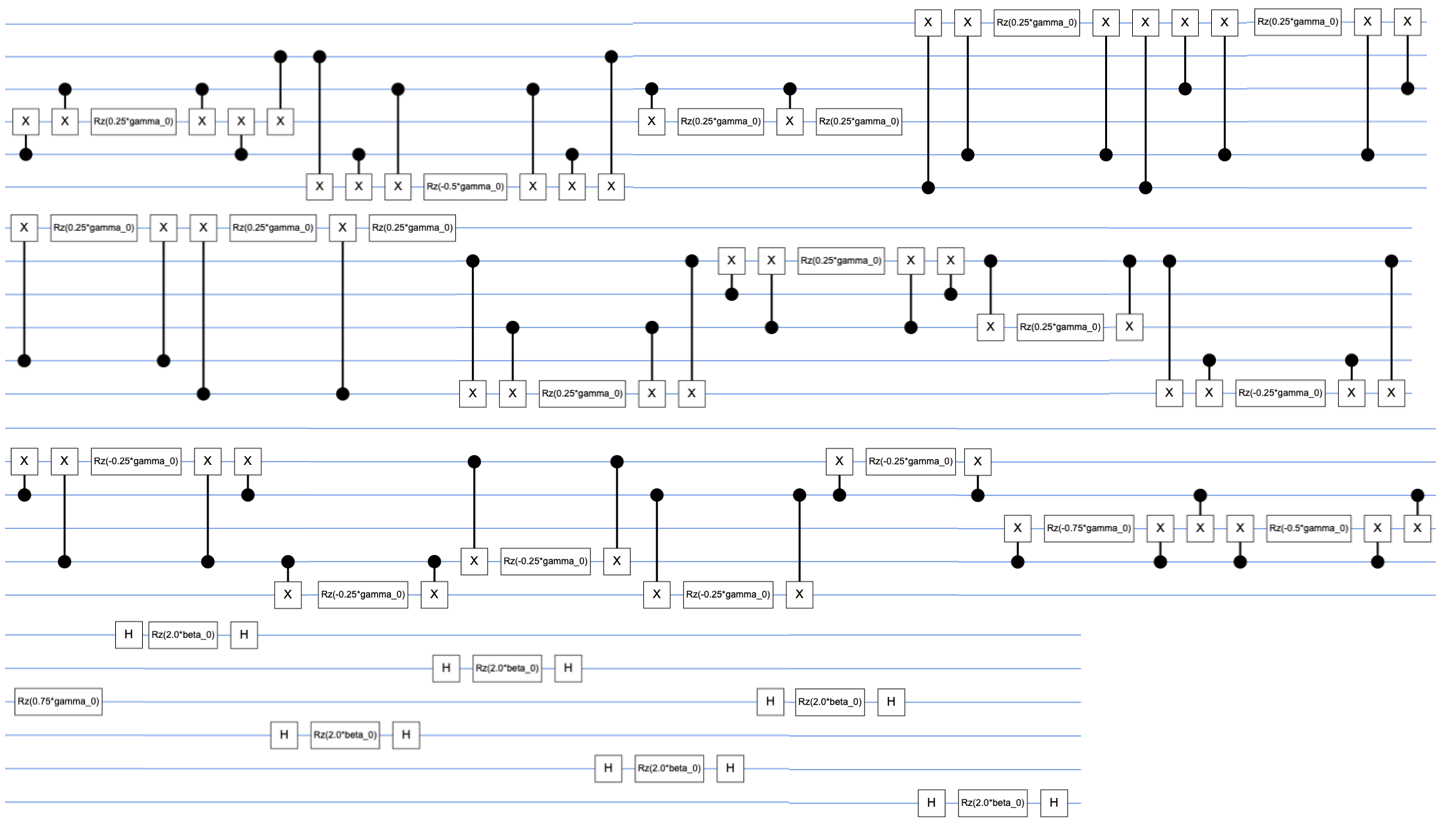}
    \label{fig:QAOA_diagram2}
   \end{minipage}
   \caption{\label{fig:QAOAcircuit}A minimal example of the QAOA circuit for the $K$-DSP for $K = 2$ and $N = 3$ with 1 qubit per qudit and fixing the circuit depth to $p =1$. This is a 6-qubit system ($K \times N \times 1$) and the circuit is hence presented in 6-line stanzas.} 
\end{figure*}
\clearpage
\nocite{*}
\bibliography{biblio}

\end{document}